\documentclass{article}
\usepackage[a4paper,margin=1in]{geometry}
\usepackage[english]{babel}
\usepackage{xcolor}
\usepackage[colorlinks]{hyperref}
\definecolor{darkred}{rgb}{0.5,0,0}
\definecolor{darkblue}{rgb}{0,0,0.5}
\definecolor{darkgreen}{rgb}{0,0.5,0}
\hypersetup{colorlinks
  ,linkcolor=darkred
  ,filecolor=darkgreen 
  ,urlcolor=darkred
  ,citecolor=darkblue,
  linktocpage}
\usepackage{amsmath, amssymb, amsthm, amsfonts}
\usepackage{mathtools}
\usepackage{centernot}
\usepackage{cleveref}
\usepackage{enumitem}
\usepackage{makecell}

\newcommand{\dotcup}{\mathbin{\dot{\cup}}}
\newcommand{\notpreceq}{\mathbin{\centernot\preceq}}

\newcommand{\symdif}{\mathbin{\triangle}}
\renewcommand*{\mod}{\mathbin{\operatorname*{mod}}}
\DeclareMathOperator{\im}{im}
\DeclareMathOperator{\dom}{dom}
\DeclareMathOperator{\qr}{qr}
\DeclareMathOperator{\free}{free}
\DeclareMathOperator{\CR}{CR}
\DeclareMathOperator{\spasm}{spasm}
\DeclareMathOperator{\sub}{sub}
\DeclareMathOperator{\cl}{cl}
\DeclareMathOperator{\NS}{NS}
\DeclareMathOperator{\ABP}{ABP}
\DeclareMathOperator{\AP}{AP}
\DeclareMathOperator{\BP}{BP}
\DeclareMathOperator{\last}{last}
\DeclareMathOperator{\leaf}{leaf}
\DeclareMathOperator{\elim}{elim}
\DeclareMathOperator{\Str}{\mathbf{Str}}
\DeclareMathOperator{\Obj}{Obj}
\DeclareMathOperator{\Ar}{Ar}
\DeclareMathOperator{\id}{id}

\theoremstyle{definition}
\newtheorem{definition}{Definition}[section]

\newtheorem{remark}[definition]{Remark}

\theoremstyle{plain}
\newtheorem{lemma}[definition]{Lemma}

\newtheorem{theorem}[definition]{Theorem}
\newtheorem{proposition}[definition]{Proposition}

\title{Homomorphism Indistinguishability and Game Comonads for Restricted Conjunction and Requantification}

\author{Georg Schindling \\ TU Darmstadt \\ \texttt{\small \href{mailto:schindling@mathematik-tu-darmstadt.de}{schindling@mathematik-tu-darmstadt.de}}}

\begin{document}

  \maketitle

  \begin{abstract}
      The notion of homomorphism indistinguishability offers a combinatorial framework for characterizing equivalence relations of graphs,
      in particular equivalences in counting logics within finite model theory.
      That is, for certain graph classes, two structures agree on all homomorphism counts from the class if and only if they satisfy the same sentences in a corresponding logic.
      This perspective often reveals connections between the combinatorial properties of graph classes and the syntactic structure of logical fragments.
      In this work, we extend this perspective to logics with restricted requantification, refining the stratification of logical resources in finite-variable counting logics.
      Specifically, we generalize Lovász-type theorems for these logics with either restricted conjunction or bounded quantifier-rank and present new combinatorial proofs of existing results.
      To this end, we introduce novel path and tree decompositions that incorporate the concept of reusability and develop characterizations based on pursuit-evasion games.
      Leveraging this framework, we establish that classes of bounded pathwidth and treewidth with reusability constraints are homomorphism distinguishing closed.
      Finally, we develop a comonadic perspective on requantification by constructing new comonads that encapsulate restricted-reusability pebble games.
      We show a tight correspondence between their coalgebras and path/tree decompositions, yielding categorical characterizations of reusability in graph decompositions.
      This unifies logical, combinatorial, and categorical perspectives on the notion of reusability.
  \end{abstract}
  \paragraph*{Funding.} The research leading to these results has received funding from from the German Research Foundation DFG (SFB-TRR 195 “Symbolic Tools in Mathematics and their Application”).
  \paragraph*{Acknowledgements.} I would like to thank Irene Heinrich, Gaurav Rattan, and Pascal Schweitzer for valuable discussions. Also, I thank the anonymous reviewers for their helpful remarks on the content and presentation of the paper. 
  \section{Introduction}
A fundamental result in graph theory due to Lovász~\cite{Lovasz1967} states that two graphs~$G$ and~$H$ are isomorphic if and only if for every graph~$F$ it holds~$\hom(F,G) = \hom(F,H)$, where~$\hom(F,G)$ denotes the number of homomorphisms from~$F$ to~$G$. 
More generally,~$G$ and~$H$ are said to be \emph{homomorphism indistinguishable} over a graph class~$\mathcal{F}$, denoted by~$G \equiv_{\mathcal{F}} H$, if~$\hom(F,G) = \hom(F,H)$ for all~$F \in \mathcal{F}$.
In this formulation, the seminal result of Lovász states that the equivalence relation~$\equiv_{\mathcal{G}}$ is the same as graph isomorphism, where~$\mathcal{G}$ is the class of all graphs.
Understanding characterizations of the homomorphism indistinguishability relation~$\equiv_{\mathcal{F}}$ not only deepens our understanding of various graph invariants but also informs algorithmic approaches to problems like graph isomorphism~\cite{Roberson24}, 
subgraph counting~\cite{Curticapean2017}, and counting answers to conjunctive queries~\cite{Goebel2024}.
In recent years, the relation~$ \equiv_{\mathcal{F}}$ was characterized for several graph classes~$\mathcal{F}$ as a natural equivalence relation arising from logic and algebra. 
Notable examples include graphs of bounded treewidth~\cite{Dell2018, Dvorak2010}, bounded pathwidth~\cite{Grohe2022, Montacute2024}, bounded tree-depth~\cite{Grohe2020, Grohe2022}, and planar graphs~\cite{Mancinska2020}.  
Two central questions have emerged in this line of research and attract ongoing interest:  
\begin{enumerate}[label={(\arabic*)}]
    \item How do structural properties of the class~$\mathcal{F}$ relate to the semantics of the relation~$\equiv_{\mathcal{F}}$? \label{q1}
\end{enumerate}  
This first question often admits elegant characterizations in terms of mathematical logic, particularly when the graph classes are defined via \emph{graph decompositions}. In this context, logical equivalence provides a natural way to express homomorphism indistinguishability. 
For a logic on graphs~$\mathsf{L}$, two graphs~$G$ and~$H$ are said to be \emph{$\mathsf{L}$-equivalent} if they satisfy exactly the same sentences from~$\mathsf{L}$.
In~\cite{Dvorak2010} Dvořák shows that two graphs are homomorphism indistinguishable over the class of graphs of treewidth at most~$k$
if and only if they are equivalent in the \emph{$k$-variable fragment}~$\mathsf{C}^k$ of \emph{first-order counting logic}~$\mathsf{C}$. 
The techniques developed in~\cite{Dvorak2010} have been refined in~\cite{Fluck2024} to prove that homomorphism indistinguishability over the class~$\mathcal{T}^k_q$
of graphs admitting tree-decompositions of \emph{width~$k$ and depth~$q$} is the same as equivalence in~$\mathsf{C}^k_q$, the fragment of~$\mathsf{C}^k$ of quantifier-rank at most~$q$.
The general technique is to directly translate between formulas in~$\mathsf{C}^k_q$ and graphs from~$\mathcal{T}^k_q$ by induction on the structure of formulas and tree decompositions in both directions.  

Beyond the characterization of homomorphism indistinguishability relations, the following natural question is fundamental to understand these relations:  
\begin{enumerate}[resume, label={(\arabic*)}]
    \item When do different graph classes~$\mathcal{F}$ induce the same indistinguishability relation~$ \equiv_{\mathcal{F}}$? \label{q2}
\end{enumerate}
This second question was approached methodically by Roberson~\cite{Roberson2022} via introducing the notion of \emph{homomorphism distinguishing closedness}.
A graph class~$\mathcal{F}$ is called \emph{homomorphism distinguishing closed} (also \emph{h.d. closed}) if for every graph~$F \notin \mathcal{F}$ there exist graphs~$G, H$ with 
$G \equiv_{\mathcal{F}} H$ and~$\hom(F, G) \neq \hom(F, H)$, i.e., if no graphs can be added to~$\mathcal{F}$ without changing the relation~$\equiv_{\mathcal{F}}$.
The significance of this notion is that any two distinct homomorphism distinguishing closed classes must induce distinct homomorphism indistinguishability relations.
In turn, equivalence relations on graphs characterized by homomorphism indistinguishability can be separated by separating the underlying graph classes, given they are homomorphism distinguishing closed.
In general, it appears to be a hard task to establish that a given class is homomorphism distinguishing closed, leading to only a short list of known examples. 
These include the class of graphs of bounded degree~\cite{Roberson2022}, bounded tree-depth~\cite{Fluck2024}, bounded treewidth~\cite{Neuen2024}, bounded depth treewidth~\cite{Adler2024}, and 
essentially profinite classes~\cite{Seppelt2024}. Beyond the investigation of specific graph classes, in~\cite{Seppelt2024} the relation between closure properties of $\mathcal{F}$
and preservation properties of $\equiv_{\mathcal{F}}$ was studied systematically guided by the two aforementioned main questions.
In~\cite{Roberson2022} Roberson conjectures that every graph class which is closed under taking disjoint unions and minors is homomorphism distinguishing closed. 
For certain such graph classes~$\mathcal{F}$, homomorphism distinguishing closedness has been successfully proven when~$\equiv_{\mathcal{F}}$ is characterized by a \emph{model-comparison game} and membership in~$\mathcal{F}$ is determined by a \emph{pursuit-evasion game}. 
A key tool in such proofs is the \emph{CFI-construction}~\cite{Cai1992}, which has been instrumental in separating the homomorphism indistinguishability relations of~$\mathsf{C}^k$-equivalence and~$\mathsf{C}^{k+1}$-equivalence. 
For a graph~$G$, the construction yields two \emph{CFI-graphs}~$X(G)$ and~$\widetilde{X}(G)$ for which their distinguishability by~$\mathsf{C}^k$ depends on the structural complexity of~$G$.    
Crucial technical challenges arise, particularly in proving the connection between pursuit-evasion and model-comparison games for the CFI-construction (see~\cite{Neuen2024}) and establishing the \emph{monotonicity} of pursuit-evasion games (see~\cite{Adler2024}). 
The monotonicity of a game ensures that when searchers have a winning strategy, the reachable positions for the evading player only decreases as the game progresses.

For many logics of interest in finite model theory, like~$\mathsf{C}^k_q$, their equivalence can be characterized in terms of model-comparison games such as
\emph{Ehrenfeucht–Fraïssé} or \emph{pebble games} (see~\cite{Ebbinghaus1995}). 
This correspondence was utilized in~\cite{Abramsky2017} to give a novel approach to logical resources in terms of \emph{game comonads}. 
The central observation is that model-comparison games induce comonads on categories of relational structures.
In this framework, several essential constructions from finite model theory can be given a categorical account, see~\cite{Abramsky2024} for a survey. 
In particular, coalgebras for some game comonads encode combinatorial parameters of structures \cite{Abramsky2021} leading to a uniform approach to homomorphism indistinguishability developed in~\cite{Dawar2021}. 
There, the first characterization of~$\mathsf{C}^k_q$-equivalence by homomorphism counts was shown for graphs admitting \emph{pebble forest covers}.
The comonadic approach was recently used to show a categorical characterization of the graph parameter pathwidth and prove that 
homomorphism indistinguishability over graphs of pathwidth at most~$k$ is logical equivalence in the \emph{restricted conjunction logic}~${\land}\mathsf{C}^k$~\cite{Montacute2024}
by building on previous work of Dalmau~\cite{Dalmau2005}.

The concept of \emph{requantification}, recently introduced in~\cite{Rassmann2025}, allows for a more refined view on stratification by logical resources in finite variable counting logics.
The logic~$\mathsf{C}^{(k_1,k_2)}$ is defined as the fragment of~$\mathsf{C}$ using at most~$k_1+k_2$ distinct variables of which only~$k_1$ may be requantified, i.e. quantified within scopes of their own quantification.
To analyze the expressive power of~$\mathsf{C}^{(k_1,k_2)}$, the \emph{bijective~$(k_1,k_2)$-pebble game}~$\BP^{(k_1,k_2)}$ and the \emph{$q$-round~$(k_1,k_2)$-cops-and-robber game}~$\CR^{(k_1,k_2)}_q$ were introduced, incorporating reusability into model-comparison and pursuit-evasion games, respectively.
In the context of Question~\ref{q1}, it is only natural to ask whether logics with restricted conjunction and requantification admit homomorphism indistinguishability characterizations. 
Subsequently, Question~\ref{q2} asks whether the corresponding graph classes are homomorphism distinguishing closed. 

\subsection*{Contribution}
In this work, we extend the study of homomorphism indistinguishability to graph classes with restricted reusability and thereby provide characterizations by counting logics with restricted requantification.  
A central contribution of our work is to show that decomposition-based techniques provide flexible tools for characterizing homomorphism indistinguishability relations, giving novel answers to Question~\ref{q1}.
Furthermore, we use these techniques for establishing homomorphism distinguishing closedness, giving new answers to Question~\ref{q2}. 
By embedding these results into the broader framework of game comonads, we provide a unified categorical perspective on requantification in finite variable logics. 
In the following, we give a more detailed description of our contribution in terms of techniques and results: 

\paragraph*{Graph decompositions.}
We answer an open question from~\cite{Rassmann2025} by characterizing the class of graphs~$\mathcal{T}^{(k_1,k_2)}_q$ where the cops have a winning strategy for~$\CR^{(k_1,k_2)}_q$ by various graph decompositions, which adapt the concept of reusability (\Cref{thm:tree:characterization}).
Furthermore, we introduce the node-searching game~$\NS^{(k_1,k_2)}$ where only~$k_1$ of the~$k_1+k_2$ searchers may be reused and characterize the class of searcher-win graphs~$\mathcal{P}^{(k_1,k_2)}$ by novel path decompositions (\Cref{thm:path:characterization}).  
We demonstrate a new effect that differentiates pathwidth from treewidth in the context of reusability.
Namely, for bounded pathwidth, non-reusable resources can be employed uniformly for the full decomposition while for bounded treewidth their usage highly depends on intermediate parts of the decomposition (\Cref{prop:search:non-reusable:first}).   
Moreover, we prove that both games~$\CR^{(k_1,k_2)}_q$ and~$\NS^{(k_1,k_2)}$ are monotone by showing that reusability is compatible with monotonicity of the non-restricted games (\Cref{prop:non-reusable:monotone}).

\paragraph*{Characterizations by logical equivalence.}
The newly introduced decompositions for the graph classes~$\mathcal{P}^{(k_1,k_2)}$ and~$\mathcal{T}^{(k_1,k_2)}_q$ form the basis of our homomorphism indistinguishability results,
providing more fine-grained answers to Question~\ref{q1}.  
By imposing constraints on requantification in the restricted conjunction logic~${\land}\mathsf{C}^{k}$ we obtain the new fragment~${\land}\mathsf{C}^{(k_1,k_2)}$
and show that~${\land}\mathsf{C}^{(k_1,k_2)}$-equivalence is exactly the same as homomorphism indistinguishability over~$\mathcal{P}^{(k_1,k_2)}$ (\Cref{thm:path:hom:ind}). 
This extends the Lovász-type theorem for~$\mathcal{P}^{(k,0)}$ from~\cite{Montacute2024} to the setting of restricted reusability. 
Interestingly, this also reproves the previous result in a purely combinatorial manner by adapting the constructive techniques from~\cite{Dvorak2010,Fluck2024}. 
We further underline the versatility of this strategy by proving that~$\mathsf{C}^{(k_1,k_2)}_q$-equivalence is homomorphism indistinguishability over~$\mathcal{T}^{(k_1,k_2)}_q$ (\Cref{thm:tree:hom:ind}). 
Also here we show how the interplay of requantification and restricted conjunction differentiates the two logics: We prove a normal form result for ${\land}\mathsf{C}_{\infty \omega}^{(k_1,k_2)}$
with respect to requantification (\Cref{prop:primitive}) which in stark contrast was ruled out for the logic $\mathsf{C}^{(k_1,k_2)}$ in \cite{Rassmann2025}.   

\paragraph*{Homomorphism distinguishing closedness.}
We utilize the established framework for counting homomorphisms from the class~$\mathcal{P}^{(k_1,k_2)}$ to prove our main technical result: 
For every graph~$G \notin \mathcal{P}^{(k_1,k_2)}$ the CFI-graphs~$X(G), \widetilde{X}(G)$ are~${\land}\mathsf{C}^{(k_1,k_2)}$-equivalent (\Cref{thm:logic:abp,lem:cfi:path:game}). 
Using a similar argument for the class~$\mathcal{T}^{(k_1,k_2)}_q$, we obtain that the classes~$\mathcal{P}^{(k_1,0)}$, the closure of $\mathcal{P}^{(k_1,k_2)}$ under disjoint unions, and~$\mathcal{T}^{(k_1,k_2)}_q$ are homomorphism distinguishing closed (\Cref{thm:pathwidth:hdc,thm:treewidth:hdc}).
This gives new answers to Question~\ref{q2} and further exemplifies the technique of using games to establish homomorphism distinguishing closedness.  
In the light of Roberson's conjecture, \Cref{thm:pathwidth:hdc} is particularly interesting as the class of graphs of pathwidth at most $k$ must exclude a fixed forest as a minor~\cite{Robertson1983}. 
Next, we employ an argument from~\cite{Neuen2024} to give an exact characterization which subgraph counts are recognized by the logics~${\land}\mathsf{C}^{(k_1,k_2)}$ and~$\mathsf{C}^{(k_1,k_2)}_q$.
For the logic~$\mathsf{C}^{(k_1,k_2)}_q$, this characterizes the ability of a reusability-restricted Weisfeiler-Leman variant to detect subgraph counts (\Cref{rem:invariance}).

\emph{Note:} Recently, \Cref{lem:cfi:path:game} and a part of its consequence \Cref{thm:pathwidth:hdc} were independently obtained for the case without constraints on reusability in the PhD thesis \cite{Seppelt2024/2025}.

\paragraph*{A comonadic perspective.}
Finally, we give a comonadic account of requantification as a logical resource. 
The \emph{pebbling comonad}~\cite{Abramsky2017} and \emph{pebble-relation comonad}~\cite{Montacute2024} were constructed from organizing the respective pebble games as endofunctors on categories of relational structures.  
We use reusability-restricted variants of these pebble games from~\cite{Rassmann2025} and this work to obtain similar constructions, namely the comonads~$\mathbb{P}^{(k_1,k_2)}$ and~$\mathbb{PR}^{(k_1,k_2)}$.
By proving close correspondences between coalgebras of these comonads and our newly defined path and tree decompositions, we obtain categorical characterizations of reusability in graph decompositions (\Cref{thm:coalgebras}).  
Finally, we show that coKleisli isomorphisms correspond to Duplicator winning strategies in the corresponding pebble games and hence characterize 
equivalence for the logics~${\land}\mathsf{C}^{(k_1,k_2)}$ and~$\mathsf{C}^{(k_1,k_2)}_q$ (\Cref{thm:isomorphism}). We also devise 
restricted pebble games to capture coKleisli morphisms and thereby characterize preservation in counting-free logics with restricted requantification (\Cref{thm:morphism}). 
  \section{Preliminaries}
We write $\mathbb{N} = \{0,1,2,\dots\}$ for the set of natural numbers, $\mathbb{N}_+ = \mathbb{N} \setminus \{0\}$ for the set of positive integers, 
and for $n\in\mathbb{N}_+$ we define $[n]\coloneqq \{1,\dots,n\}$.
Unless stated explicitly otherwise, we let $k_1,k_2 \in \mathbb{N}$ throughout the paper. 
We fix the \emph{variable sets}  (also called \emph{pebble sets})~$[x_{k_1}] \coloneqq \{x_1,\dotsc,x_{k_1}\}$, $[y_{k_2}] \coloneqq \{y_1,\dotsc,y_{k_2}\}$, and~$[x_{k_1}, y_{k_2}] \coloneqq \{x_1,\dotsc,x_{k_1},y_1,\dotsc,y_{k_2}\}$.
For the following definitions we let $V$ be a set. We write $2^V$ for the \emph{power set} of $V$ and set $\binom{V}{2} = \{U \in 2^V : |U|=2\}$.
A \emph{partial function} $\alpha \colon [x_{k_1}, y_{k_2}] \rightharpoonup V$ assigns to every variable $z \in [x_{k_1}, y_{k_2}]$ at most one element $\alpha(z) \in V$.
If $\alpha$ does not assign an element to $z$, we write $\alpha(z) = \bot$. Also, we write $\im(\alpha)$ and $\dom(\alpha)$ for the \emph{image} and \emph{domain} of $\alpha$ respectively.
We write $V^+$, $V^n$ and $V^{\leq n}$ for the sets of non-empty finite sequences, sequences of length $n$, and sequences of length at most $n$ over $V$ respectively. 
We denote sequences of elements $s_1,\dotsc,s_n \in V$ by $\overline{s} = [s_1,\dotsc,s_n] \in V^n$ and for $\overline{s}, \overline{t} \in V^+$ we write $\overline{s} \sqsubseteq \overline{t}$ if $\overline{s}$ is a \emph{prefix} of $\overline{t}$.
The \emph{concatenation} of $\overline{s}$ and $\overline{t}$ is denoted by $\overline{s}\overline{t}$. For $i,j \in [n]$ with $i \leq j$ we define $\overline{s}[i,j] \coloneqq [s_i,s_{i+1},\dots,s_{j}]$ and $\overline{s}[i] \coloneqq s[i,i]$. 
Also, we indicate that $s_i$ occurs in $\overline{s}$ by writing $s_i \in \overline{s}$. 
For a variable $z \in [x_{k_1}, y_{k_2}]$ and $v \in V$ we write $\alpha[z/v]$ for the partial function that is obtained from $\alpha$ by replacing the image $\alpha(z)$ by $v$.
Given a sequence $\overline{s} = [(z_1,v_1), \dotsc, (z_n,v_n)] \in ([x_{k_1}, y_{k_2}] \times V)^n$ and $z \in [x_{k_1}, y_{k_2}]$ we denote by $\last_z(\overline{s})$ the last $v_i$ such that $(z,v_i) \in \overline{s}$.
We call the first entry of each element in $\overline{s}$ a \emph{pebble index} or \emph{variable index}. 
For a proposition $P$, we use the \emph{Iverson bracket} $[P] \in \{0,1\}$ to indicate whether $P$ is satisfied.

\paragraph*{Finite model theory.}
We fix a finite signature $\sigma$ of relation symbols and associate to each $R \in \sigma$ an \emph{arity} $\operatorname{ar}(R) \in \mathbb{N}_+$.
A \emph{$\sigma$-structure} $\mathcal{A}$ consists of a \emph{universe} of elements $V(\mathcal{A})$ and interpretations $R^{\mathcal{A}} \subseteq V(\mathcal{A})^{\operatorname{ar}(R)}$ for each $R \in \sigma$.
The \emph{Gaifman graph} of a $\sigma$-structure $\mathcal{A}$ is defined as 
$$G(\mathcal{A}) \coloneqq (V(\mathcal{A}), \{\{a,b\} : a,b \in V(\mathcal{A}), a \neq b, \text{and there exist } R \in \sigma, \overline{c} \in R^{\mathcal{A}} \text{ with } a,b \in \overline{c}\})$$
For $\sigma$-structures $\mathcal{A}$ and $\mathcal{B}$ we say that $\mathcal{B}$ is a \emph{substructure} of $\mathcal{A}$ if $V(\mathcal{B}) \subseteq V(\mathcal{A})$ and $R^{\mathcal{B}} \subseteq R^{\mathcal{A}}$ for each $R \in \sigma$.
Every set of elements $A \subseteq V(\mathcal{A})$ induces a substructure of $\mathcal{A}$ with universe $A$ and relations $R^\mathcal{A} \cap A^{\operatorname{ar}(R)}$ for $R \in \sigma$.
A \emph{homomorphism} between $\sigma$-structures $\mathcal{A}$ and $\mathcal{B}$ is a function $h \colon V(\mathcal{A}) \to V(\mathcal{B})$ such that for all $R \in \sigma$
we have that $(v_1,\dotsc,v_{\operatorname{ar}(R)}) \in R^{\mathcal{A}}$ implies $(h(v_1),\dotsc,h(v_{\operatorname{ar}(R)}))\in R^{\mathcal{B}}$.
The function $h$ is called an \emph{isomorphism} if it is a bijective homomorphism and $h^{-1}$ is a homomorphism.
Let $\mathsf{L}$ be a logic over the signature $\sigma$ with variable set $\mathcal{V}$ (for a formal definition see \cite{Ebbinghaus2017}).
For a formula $\varphi \in \mathsf{L}$ we write $\varphi(v_1,\dots,v_n)$ to indicate that the set of \emph{free variables} of $\varphi$, which we denote by $\free(\varphi)$, is a subset of $\{v_1,\dots,v_n\}$.
Given a $\sigma$-structure $\mathcal{A}$ and an assignment $\alpha \colon \mathcal{V} \rightharpoonup V(\mathcal{A})$ we write $\mathcal{A}, \alpha \models \varphi$
to indicate that $\mathcal{A}$ satisfies $\varphi$ with $\free(\varphi)$ interpreted according to $\alpha$.
For a tuple $\overline{a}\in V(\mathcal{A})^n$ we may write $\mathcal{A}, \overline{a} \models \varphi$ by assigning $v_i \mapsto a_i$.
For $\sigma$-structures $\mathcal{A}, \mathcal{B}$ we write $A \equiv_{\mathsf{L}} \mathcal{B}$ if  $\mathcal{A}$ and $\mathcal{B}$ satisfy exactly the same sentences from $\mathsf{L}$.
We write $A \Rrightarrow_{\mathsf{L}} \mathcal{B}$ if every sentence from $\mathsf{L}$ satisfied by $\mathcal{A}$ is also satisfied by $\mathcal{B}$.
\emph{First-order counting logic} $\mathsf{C}$ extends first-order logic $\mathsf{FO}$ by counting quantifiers $\exists^{\geq n} \varphi$ for $\varphi \in \mathsf{C}$.
We say that a variable $x_i$ is \emph{requantified} in a logical formula if it either occurs free and bound or if it is quantified within the scope of a quantification over $x_i$.
The logic $\mathsf{C}^{(k_1,k_2)}$ is obtained from $\mathsf{C}$ by fixing the variable set $[x_{k_1}, y_{k_2}]$ and requiring
that only variables from $[x_{k_1}]$ are requantified. Finally, $\mathsf{C}^{(k_1,k_2)}_q$ is the fragment of $\mathsf{C}^{(k_1,k_2)}$ with quantifier-rank at most $q$ (see \cite{Rassmann2025} for details).

\paragraph*{Graphs.}
A finite \emph{graph} is a pair \(G=(V(G),E(G))\) consisting of a finite set \(V(G)\) of \emph{vertices} and a set \(E(G)\subseteq\binom{V(G)}{2}\) of \emph{edges}.
For an edge $\{u,v\} \in E(G)$ we also write $uv \in E(G)$. Given a set $W \subseteq V(G)$ we define the \emph{induced subgraph} $G[W] = (W, \{uv \in E(G): u,v \in W\})$.
For a graph $G$ and $v\in V(G)$ we write $E(v)$ for the set of edges incident to $v$.
We denote the set of vertex sets of \emph{connected components} (i.e. maximal connected subgraphs) of $G$ by $\mathcal{C}_G$.
A \emph{rooted tree} is a pair $(T,r)$ such that $T$ is a tree and $r \in V(T)$ is a designated vertex, called the \emph{root} of $T$.
With a rooted tree we associate a partial order $\preceq_T$ on the vertices of $T$ by setting $s \preceq_T t$ exactly if $s$ is on the unique path from $r$ to $t$.
The \emph{height} of a rooted tree is the maximal number of vertices on a path from the root to a leaf. 
We denote the set of leaves of a tree $T$ by $\leaf(T)$. 
A \emph{rooted forest} is a pair $(F, \overline{r})$ such that if $F_1,\dotsc,F_p$ are the connected components of $F$ we have $\overline{r} = (r_1,\dotsc,r_p)$ and $(F_i, r_i)$ is a rooted tree for each $i \in [p]$.
A \emph{labeled graph} is a graph $G$ together with a finite set of \emph{labels} $L$ and a partial labeling function $\nu_G \colon L \rightharpoonup V(G)$.
We write $L_G \coloneqq \dom(\nu)$ for the set of labels occurring in $G$.
A labeled graph is called \emph{fully labeled} if the labeling function is surjective.
We denote the class of all $[x_{k_1}, y_{k_2}]$-labeled graphs by $\mathcal{G}_{[x_{k_1}, y_{k_2}]}$.
For a label $z \in [x_{k_1}, y_{k_2}]$ and $v \in V(G) \cup \{\bot\}$ we define the labeled graph $G[z \to v]$ as $G$ with labeling function $\nu_{G[z \to v]} \coloneqq \nu_G[z / v]$.
For a tuple of labels $\overline{z} \in [x_{k_1}, y_{k_2}]^{n}$ and $\overline{v} \in (V(G) \cup \{\bot\})^n$ we inductively define
$G_1 \coloneqq G[z_1 \to v_1]$ and $G_{i+1} \coloneqq G_{i}[z_{i+1} \to v_{i+1}]$ for $i \in [n-1]$ and set $G[\overline{z} \to \overline{v}] \coloneqq G_n$.
The product of two labeled graphs $G_1, G_2 \in \mathcal{G}_{[x_{k_1}, y_{k_2}]}$ is the graph $G_1G_2$
obtained by taking the disjoint union of $G_1$ and $G_2$, identifying vertices with the same label,
and suppressing any loops or parallel edges that might be created.
Note that for a graph $G \in \mathcal{G}_{[x_{k_1}, y_{k_2}]}$ the labeling $\nu_G$ is a partial variable assignment and hence we may write $G \models \varphi$ for $G, \nu_G \models \varphi$
if $\free(\varphi) \subseteq L_G$.

Two central concepts in this work are \emph{tree decompositions} and \emph{path decompositions}, which we briefly introduce next. For a more detailed exposition, we refer the reader to \cite{Bodlaender1998}.
\begin{definition}
  Let $G$ be a graph.
  A \emph{tree decomposition} of $G$ is a tuple $(T,\beta)$ such that $T$ is a tree and $\beta \colon V(T) \to 2^{V(G)}$ is a function such that
  \begin{itemize}
    \item $\bigcup_{t \in V(T)} \beta(t) = V(G)$,
    \item for all $uv \in E(G)$ there exists a $t \in V(T)$ with $u,v \in \beta(t)$, and
    \item for all $v \in V(G)$ the set of vertices $\beta^{-1}(\{v\}) = \{t \in V(T): v \in \beta(t)\}$ is connected in $T$.
  \end{itemize}
  The \emph{width} of $(T, \beta)$ is $\max_{t \in V(T)} |\beta(t)|-1$ and the \emph{treewidth} of $G$ is the minimal width of a tree decomposition of $G$.
  If $T$ is a rooted tree with root $r \in V(T)$ we call $(T,r,\beta)$ a \emph{rooted tree decomposition}.
  A \emph{path decomposition} of $G$ is a tree decomposition $(P, \beta)$ such that the underlying tree $P$ is a path.
  If the decomposition is rooted, then we define the root to be an endpoint of $P$. 
  The \emph{width} of $(P,\beta)$ is again $\max_{p \in V(P)} |\beta(p)|-1$ and the \emph{pathwidth} of $G$ is the minimal width of a path decomposition of $G$.
\end{definition}

\paragraph*{The CFI-construction.}
We use the \emph{CFI-construction} introduced in \cite{Cai1992} in the variant presented in \cite{Roberson2022}.
Let $G$ be a connected graph and $U \subseteq V(G)$. For each $v \in V(G)$ we set $\delta_{v, U} \coloneqq |\{v\} \cap U|$.
The \emph{CFI-graph} $X_U(G)$ is defined by
\begin{align*}
  V(X_U(G)) & \coloneqq \{(v,S) : v \in V(G), S \subseteq E(v), |S| \equiv \delta_{v, U} \mod 2\}, \\
  E(X_U(G)) & \coloneqq \{(v,S)(u,T) : uv \in E(G), uv \notin S \symdif T\}.
\end{align*}
The connected graph $G$ is called the \emph{base graph} of $X_U(G)$. We also define  $\rho \colon V(X_U(G)) \to G, (u, S) \mapsto u$.
For $v \in V(G)$ we also denote $F_U(v) \coloneqq \{(v,S) : S \subseteq E(v),  |S| \equiv \delta_{v, U} \mod 2\}$ for the vertices in $X_U(G)$ associated with $v$.
These vertices are also referred to as \emph{gadget vertices} of $v$ and $F_U(v)$ as the corresponding \emph{gadget}.

\begin{lemma}[{\cite[Lemma 6.2]{Cai1992}}, {\cite[Lemma 3.2]{Roberson2022}}] \label{lem:cfi:isomorphic}
  For all sets of base vertices $S, T \subseteq V(G)$, the graphs $X_T(G)$ and $X_S(G)$ are isomorphic if and only if $|S| \equiv |T| \mod 2$.
\end{lemma}

Thus, we set $X(G) \coloneqq X_\emptyset(G)$ and $\widetilde{X}(G) \coloneqq X_{\{v\}}(G)$ for some $v \in V(G)$ as the isomorphism type only depends on the parity of $|U|$.
Note that the vertex sets of the graphs $X(G)$ and $\widetilde{X}_u(G)$ only differ in $F_{\emptyset}(u)$ and $F_u(u)$.

\begin{lemma}[{\cite[Lemma 11]{Neuen2024}}] \label{lem:cfi:twist:iso}
  Let $G$ be a connected graph, $u,v \in V(G)$, and $P$ be a path from $u$ to $v$ in $G$.
  Then there exists an isomorphism $\varphi_{u,v} \colon X_{\{u\}}(G) \to X_{\{v\}}(G)$ such that for all $(w, S) \in V(X_{\{u\}}(G))$ it holds that
  \begin{enumerate}
    \item $\rho(\varphi_{u,v}(w,S)) = w$, and
    \item if $w \in V(G) \setminus P$ then $\varphi_{u,v}(w,S) = (w, S)$.
  \end{enumerate}
\end{lemma}

\paragraph*{Homomorphism indistinguishability.}
We denote the number of homomorphisms from a possibly labeled graph $F$ to a graph $G$ by $\hom(F,G)$.
For a class of labeled graphs $\mathcal{F}$ we write $\mathbb{R}\mathcal{F}$ for the class of all formal finite linear combinations with real coefficients of graphs in $\mathcal{F}$.
For a linear combination $\mathfrak{F} = \sum_i c_i F_i \in \mathbb{R}\mathcal{F}$ and a labeled graph $G$ we define $\hom(\mathfrak{F}, G) \coloneqq \sum_i c_i \hom(F_i, G)$ and $L_{\mathfrak{F}}$ to be the set of labels occurring in $\mathfrak{F}$.
We collect some basic facts about $\hom(\cdot, G)$:
\begin{lemma}\label{lem:hom:counts}
  Let $F,G \in \mathcal{G}_{[x_{k_1},y_{k_2}]}$ with $L_F \subseteq L_G$. Then the following assertions hold:
  \begin{enumerate}
    \item If $F$ is fully labeled, then $\hom(F,G)=1$ if for all $z_1,z_2 \in L_F$ we have
          \begin{itemize}
            \item $\nu_F(z_1) = \nu_F(z_2) \Longrightarrow \nu_G(z_1) = \nu_G(z_2)$, and
            \item $\nu_F(z_1)\nu_F(z_2) \in E(F) \Longrightarrow \nu_G(z_1)\nu_G(z_2) \in E(G)$.
          \end{itemize}
          Otherwise, it holds that $\hom(F,G)=0$.
    \item If $F = F_1F_2$ for $F_1,F_2 \in \mathcal{G}_{[x_{k_1}, y_{k_2}]}$, then $\hom(F,G) = \hom(F_1,G)\hom(F_2,G)$
    \item If $z \in L_{F}$ and $F' = F[z \to \bot]$, then $\hom(F',G) = \sum_{v \in V(G)} \hom(F, G[z \to v])$
  \end{enumerate}
\end{lemma}

For a class $\mathcal{F}$ of graphs we say that two graphs $G$ and $H$ are \emph{homomorphism indistinguishable over $\mathcal{F}$} if for all $F \in \mathcal{F}$ it holds that $\hom(F,G) = \hom(F,H)$.
In this case we write $G \equiv_\mathcal{F} H$.
The \emph{homomorphism distinguishing closure} of $\mathcal{F}$ is defined as the class
$$\cl(\mathcal{F}) \coloneqq \{F \in \mathcal{G}: \forall G, H \in \mathcal{G} \quad G \equiv_{\mathcal{F}} H \Rightarrow \hom(F,G) = \hom(F, H)\}.$$
The class $\mathcal{F}$ is called \emph{homomorphism distinguishing closed} (also \emph{h.d. closed}) if $\cl(\mathcal{F}) = \mathcal{F}$.
We will use the following sufficient condition for h.d. closedness of graph classes satisfying certain closure conditions. 
\begin{lemma}[{\cite[Proposition 47]{Fluck2023}}] \label{lem:sufficient:hdc}
  Let $\mathcal{F}$ be a graph class that is closed under taking disjoint unions and summands (i.e. $F_1 \dotcup F_2 \in \mathcal{F}$ exactly if $F_1, F_2 \in \mathcal{F}$).
  If for every connected graph $G \notin \mathcal{F}$ it holds that $X(G) \equiv_{\mathcal{F}} \widetilde{X}(G)$, then $\mathcal{F}$ is homomorphism distinguishing closed.
\end{lemma}

\paragraph*{Category theory.}  We assume only very basic background in category theory, see \cite{Abramsky2010} for details. 
For a category $\mathbf{C}$ we denote its objects by $\Obj(\mathbf{C})$ and its morphisms (or arrows) by $\Ar(\mathbf{C})$.
We denote the category of $\sigma$-structures with their homomorphisms by~$\Str(\sigma)$.
For a functor $J \colon \mathbf{J} \to \mathbf{C}$ a \emph{relative comonad on $J$} is given by
\begin{itemize}
  \item an \emph{object map} $\mathbb{G} \colon \Obj(\mathbf{J}) \to \Obj(\mathbf{C})$,
  \item a \emph{counit morphism} $\varepsilon_{\mathcal{A}} \colon \mathbb{G}\mathcal{A} \to J\mathcal{A}$ for every $\mathcal{A} \in \Obj(\mathbf{J})$,
  \item and a \emph{coextension operation} $(\cdot)^\ast$ associating with each morphism $f \colon \mathbb{G}\mathcal{A} \to J\mathcal{B}$ another morphism
        $f^\ast \colon \mathbb{G}\mathcal{A} \to \mathbb{G}\mathcal{B}$ for $\mathcal{A}, \mathcal{B} \in \Obj(\mathbf{J})$
\end{itemize}
such that for all morphisms $f \colon \mathbb{G}\mathcal{A} \to J\mathcal{B}$, $g \colon \mathbb{G}\mathcal{B} \to J\mathcal{C}$ we have
$\varepsilon^\ast_{\mathcal{A}} = \id_{\mathbb{G}\mathcal{A}}$, $\varepsilon_{\mathcal{B}} \circ f^\ast = f$, and $(g \circ f^\ast)^\ast = g^\ast \circ f^\ast$.

For the relative comonad $(\mathbb{G}, \varepsilon, (\cdot)^\ast)$ we define the associated \emph{coKleisli category} $\mathcal{K}(\mathbb{G})$: 
\begin{itemize}
  \item $\Obj(\mathcal{K}(\mathbb{G}))$ is the class of objects $\Obj(\mathbf{J})$.
  \item $\Ar(\mathcal{K}(\mathbb{G}))$ are all morphisms $f \colon \mathbb{G} \mathcal{A} \to J \mathcal{B}$ for $\mathcal{A}, \mathcal{B} \in \Obj(\mathbf{J})$.
  \item The composition $\circ_{\mathcal{K}}$ is defined by setting $g \circ_{\mathcal{K}} f \coloneqq g \circ f^\ast$.
  \item The identity morphisms are given by the counit $\varepsilon_{\mathcal{A}}$ for $\mathcal{A} \in \Obj(\mathbf{J})$. 
\end{itemize}
The notion of a \emph{comonad in coKleisli form} and the corresponding coKleisli category can be recovered when $\mathbf{J} = \mathbf{C}$ and $J = \id_{\mathbf{C}}$.
From this, a \emph{comonad in standard form} $(\mathbb{G}, \varepsilon, \delta)$ on the category $\mathbf{C}$ can be obtained 
by setting $\mathbb{G}f \coloneqq (f \circ \varepsilon_\mathcal{A})^\ast$ (turning $\mathbb{G}$ into a functor) and $\delta_{\mathcal{A}} \coloneqq \id_{\mathbb{G}\mathcal{A}}^\ast$ for $\mathcal{A} \in \Obj(\mathbf{C})$. 
For a comonad in standard form, a \emph{coalgebra over $\mathbb{G}$} is a pair $(\mathcal{A}, \alpha)$ where $\mathcal{A} \in \Obj(\mathbf{C})$ and $\Ar(\mathbf{C}) \ni \alpha \colon \mathcal{A} \to \mathbb{G}\mathcal{A}$ such that 
$\delta_{\mathcal{A}} \circ \alpha = \mathbb{G}\alpha \circ \alpha$ and $\varepsilon_{\mathcal{A}} \circ \alpha = \id_{\mathcal{A}}$.
  \section{Graph decompositions with restricted reusability} \label{sec:decompositions}

In this section, we introduce several graph decompositions that incorporate constraints on reusability of vertices within the parts of each decomposition. 
Specifically, we define four distinct types of decompositions, each of which comes in two variants corresponding to two underlying structural models:
the \emph{path model} and the \emph{(bounded-depth) tree model}. 
\begin{enumerate}
  \item \emph{path decompositions with exception sets} and \emph{tree decomposition with exception sets}
  \item \emph{linear pebble forest covers} and \emph{pebble forest covers}
  \item \emph{node searching games} and \emph{cops-and-robber games}
  \item \emph{construction caterpillars} and \emph{construction trees}
\end{enumerate} 
Each decomposition is parameterized by two values: $k_1$, representing the number of \emph{reusable resources}, and $k_2$, representing the number of \emph{non-reusable resources}.
Our goal is to show that, for each of the two models, the corresponding decompositions define the same graph class with aligned parameters. This equivalence allows us to explore the concept of reusability in graph decompositions from multiple, yet consistent, perspectives.
The motivation for each decomposition arises from its significance in homomorphism indistinguishability over the associated graph class, establishing a unifying theme for our study.
In \Cref{fig:table} we summarize the significance of the various decompositions introduced in this section for our results on homomorphism indistinguishability.

\begin{figure}[h]
  \centering
  \resizebox{\textwidth}{!}{
    \begin{tabular}{ c||c|c|c|c } 
      ~ & graph class & \makecell{homomorphism \\ indistinguishability} & h.d. closedness & game comonads \\
      \hline
      \textbf{path} & 
      \makecell{$\mathcal{P}^{(k_1,k_2)}$ defined by \\ path decompositions} & 
      \makecell{${\land}\mathsf{C}^{(k_1,k_2)}$-equivalence \\ via construction caterpillars \\ $\mathcal{LP}^{(k_1,k_2)}$} & 
      \makecell{only for ${\dotcup}\mathcal{P}^{(k_1,k_2)}$ \\ via $\NS^{(k_1,k_2)}$} & 
      \makecell{coalgebras over $\mathbb{PR}_{(k_1,k_2)}$ \\ correspond to linear \\ component forest covers} \\
      \hline
      \textbf{tree} & 
      \makecell{$\mathcal{T}^{(k_1,k_2)}_q$ defined by \\ bounded depth \\ tree decompositions} & 
      \makecell{$\mathsf{C}^{(k_1,k_2)}_q$-equivalence \\ via construction trees \\ $\mathcal{LT}^{(k_1,k_2)}_q$} & 
      \makecell{true via \\ $\CR^{(k_1,k_2)}_q$} & 
      \makecell{coalgebras over $\mathbb{P}_{(k_1,k_2)}^q$ \\ correspond to \\ forest covers}
    \end{tabular}
  }
  \caption{Utilization of graph decompositions for homomorphism indistinguishability and game comonads}
  \label{fig:table}
\end{figure}

We begin by introducing the notion of \emph{exception sets} for tree decompositions of bounded depth and for path decompositions. 
The key idea is to generalize the concept of a tree decomposition of width $k_1$ by permitting up to $k_2$ \emph{exceptions} along each branch of the underlying tree. 
However, there is an important restriction: if a vertex is designated as an exception at some node $t$ in the rooted tree decomposition, then it must remain fixed, i.e., cannot be replaced by a different vertex at any descendant of $t$. 
In this sense, the exception status is not reusable along the subtree rooted at $t$.

\begin{definition}
  Let $G$ be a graph, $k_2 \in \mathbb{N}$, $k_1 \in \mathbb{N} \cup \{-[k_2>0]\}$, and $k_1+k_2, q \in \mathbb{N}_+$. 
  A rooted tree decomposition~$(T, r, \beta)$ of $G$ has 
  \begin{itemize}
    \item \emph{width} $(k_1,k_2)$ if for each leaf $\ell \in V(T)$ there exists a set 
    of \emph{exceptions} $S_{\ell} \subseteq V(G)$ with $|S_{\ell}| \leq k_2$ such that~$\max_{t \preceq_T \ell} |\beta(t) \setminus S_{\ell}|-1 \leq k_1$,
    \item and \emph{depth} $\max_{v \in V(T)}|\bigcup_{t \preceq_T v} \beta(t)|$.
  \end{itemize}
  We write $\mathcal{T}^{(k_1+1,k_2)}$ for the class of all graphs admitting a tree decomposition of width $(k_1,k_2)$ 
  and $\mathcal{T}^{(k_1+1,k_2)}_q$ for the subclass admitting such a cover of depth $q$.

  If $(P, r, \beta)$ is a path decomposition, we also define that $(P,r, \beta)$ has \emph{component width} $(k_1,k_2)$ if for each connected component $C \subseteq V(G)$ 
  there is an exception set $S_C \subseteq C$ with $|S_C| \leq k_2$ such that $\max_{p \in V(P)} |\beta(p) \setminus \bigcup_{C \in \mathcal{C}_G}S_C|-1 \leq k_1$. 

  We write $\mathcal{P}^{(k_1+1,k_2)}$ for the class of all graphs admitting a path decomposition of width $(k_1,k_2)$.
\end{definition}
For $k_2=0$ we recover the class $\mathcal{T}^{(k_1+1,0)}$ of graphs of treewidth at most $k_1$ but for $k_2>0$ we allow the technical nuisance of having $k_1=-1$ in order to avoid a case distinction. 

Note that for the \emph{width} of a path decomposition we only require the existence of one single set $S$ of exceptions.
Thus, for fixed $k_1$ the class of graphs admitting a path decomposition of width $(k_1,k_2)$ can be seen as an approximation of the class of graphs of pathwidth at most $k_1$ up to deleting $k_2$ vertices. 

For the proofs of our characterization theorems, we use the well-known notion of \emph{nice} tree decompositions (see \cite{Bodlaender1998}) and show that they are compatible with reusability. 

\begin{definition}
    Let $(T, r, \beta)$ be a rooted tree decomposition of a graph $G$. A node $t \in V(T)$ is called
    \begin{itemize}
        \item \emph{introduce node} if it has exactly one child $s$ such that there exists a vertex $v \in V(G)$ with the property $\beta(t) = \beta(s) \cup \{v\}$,
        \item \emph{forget node} if it has exactly one child $s$ such that there exists a vertex $v \in V(G)$ with the property $\beta(s) = \beta(t) \cup \{v\}$, or
        \item \emph{join node} if it has exactly two children $s_1,s_2$ with $\beta(t) = \beta(s_1) = \beta(s_2)$.
    \end{itemize}
    A rooted tree decomposition is called \emph{nice} if every node that is not a leaf is either an introduce, join, or forget node.
    A rooted path decomposition is called \emph{nice} if every node that is not the endpoint  of the underlying path is an introduce or forget node.
\end{definition}

\begin{lemma} \label{lem:nice:decomp}
    Let $G$ be a graph with a tree decomposition $(T, r, \beta)$ of width~$(k_1,k_2)$ and depth~$q$.
    Then $G$ has a nice tree decomposition $(T', r', \beta')$ of width $(k_1,k_2)$ and depth~$q$.
    Furthermore, each leaf $\ell \in V(T')$ the set $S_{\ell}$ can be chosen such that for each node $t \in V(T')$ and all children $t'$ of $t$ it holds that $S_{\ell} \cap \beta(t) \subseteq \beta(t')$.
    The same holds for path decompositions of width~$(k_1,k_2)$.
\end{lemma}
\begin{proof}
    Given a tree decomposition $(T, r, \beta)$ of $G$ of width $(k_1,k_2)$ and depth $q$ we construct a nice decomposition of the same width and depth as follows:

    First, make the decomposition \emph{non-redundant} by contracting all edges $st \in E(T)$ with $\beta(s) \subseteq \beta(t)$ to a single bag.
    We replace each node $t \in V(T)$ with degree $d \geq 3$ by a binary tree $B_t$ with $d$ leaves.
    For each of the nodes $t' \in V(B_t)$ we set $\beta(t') \coloneqq \beta(t)$.
    Then each of the subtrees induced by the children of $t$ in $T$ is attached to exactly one of the laves of $B_t$.
    Now let $t \in V(T)$ such that its parent $p(t)$ has degree $2$.
    We replace the edge $p(t)t \in E(T)$ by a path of length $|\beta(p(t)) \mathbin{\triangle} \beta(t)|$, on which we first forget the nodes from $\beta(p(t)) \setminus \beta(t)$ one by one and afterwards introduce the nodes from $\beta(t) \setminus \beta(p(t))$.

    To obtain a nice path decomposition of the same width it suffices to make it non-redundant and then use the forget-and-introduce technique.
\end{proof}

Next, we extend the notion of \emph{pebble forest covers} to incorporate the concept of reusability. 
Originally introduced in \cite{Abramsky2017} as \emph{$k$-traversals}, these structures provide a combinatorial characterization of coalgebras over the pebbling comonad $\mathbb{P}_k$, finally demonstrating how the comonadic structure of $\mathbb{P}_k$ can be used to characterize treewidth. 
Our aim is to follow a similar approach to characterize reusability (specifically, width $(k_1, k_2)$) within tree decompositions of bounded depth via comonadic methods. 
To this end, we use \emph{non-reusable pebbles} $y_1, \dotsc, y_{k_2}$, which mark fixed positions in a forest cover that cannot be reassigned, thereby encoding the non-reusability constraint directly into the structure.

\begin{definition}
  Let $G$ be a graph and $k_1+k_2 \in \mathbb{N}_+$.
  A \emph{$(k_1,k_2)$-pebble forest cover} of $G$ is a tuple $(F, \overline{r}, p)$ where $(F, \overline{r})$ is a rooted forest with $V(F) = V(G)$
  and $p \colon V(G) \to [x_{k_1}, y_{k_2}]$ is a pebbling function such that
  \begin{enumerate}
    \item if $uv \in E(G)$, then $u \preceq_F v$ or $v \preceq_F u$,
    \item if $uv \in E(G)$ and $u \prec_F v$, then for all $w \in V(G)$ with $u \prec_F w \preceq_F v$ it holds that $p(u) \neq p(w)$, and
    \item if $u \in V(G)$ and $p(u) \in [y_{k_2}]$, then for all $w \in V(G)$ with $u \prec_F w$ it holds that~$p(u) \neq p(w)$.
  \end{enumerate}
  The forest cover $(F, \overline{r}, p)$ has \emph{depth} $q \in \mathbb{N}_+$ if $(F, \overline{r})$ has height $q$.
  We call $(F, \overline{r}, p)$ a \emph{linear forest cover} if it additionally holds that
  \begin{enumerate}
    \item every connected component of $F$ is a path, and
    \item \label{lfc:non-reusable} if $u \in V(G)$ and $p(u) \in [y_{k_2}]$, then for every $w \in V(G) \setminus \{u\}$ it holds that~$p(u) \neq p(w)$.
  \end{enumerate}
  If we relax \Cref{lfc:non-reusable} such that for $u \in V(G)$ with $p(u) \in [y_{k_2}]$ only for every $w$ in the same path of $F$ as $u$ it must hold~$p(u) \neq p(w)$,
  we say that $(F, \overline{r}, p)$ is a \emph{linear component forest cover}. 
\end{definition}

The \emph{tree-depth} of a graph $G$ is the minimum $q \in \mathbb{N}_+$ such that $G$ admits a forest cover of depth at most $q$~\cite{Nesetril2006}. The class of graphs of tree-depth at most $q$ is denoted by $\mathcal{T}_q$. 

The classes $\mathcal{P}^{(k,0)}$, $\mathcal{T}^{(k,0)}$, and $\mathcal{T}_q$ can be characterized in terms of pursuit-evasion games by \cite{Seymour1993}, \cite{Bienstock1991, Kirousis1985}, and \cite{Giannopoulou2012} respectively.
The characterization of $\mathcal{T}^{(k,0)}$ was recently refined to hold for~$\mathcal{T}^{(k,0)}_q$ in~\cite{Fluck2024}. 
In \cite{Rassmann2025} a cops-and-robber game with constraints on the reusability of cops and the number of rounds was introduced. 
We recall its definition and also modify it to match the game-theoretic characterization of $\mathcal{P}^{(k,0)}$. 

\begin{definition}
  Let $G$ be a graph and let $k_1+k_2,q \in \mathbb{N}_+$.
  The \emph{cops-and-robber game} $\CR^{(k_1,k_2)}_q(G)$ is defined as follows:
  The game is played between a group of $k_1+k_2$ \emph{cops} denoted by the elements in $[x_{k_1}, y_{k_2}]$ and one \emph{robber}.
  The position of the cops is given by a function $\gamma \colon [x_{k_1}, y_{k_2}] \rightharpoonup V(G)$ and the position of the robber is a vertex $v \in V(G)$.
  We denote the connected component of $v$ in the graph $G - \im(\gamma)$ by~$C^\gamma_v$.
  In one round of the game, the following steps are performed
  \begin{enumerate}
    \item The cops choose one cop $z \in [x_{k_1}] \cup \{y \in [y_{k_2}] : \gamma(y) = \bot\}$ and declare a destination $w \in V(G) \cup \{\bot\}$.
    \item The robber chooses a vertex $v'$ in $C^{\gamma[z / \bot]}_v$.
    \item If $v' \in \im(\gamma[z/w])$ the cops win. Otherwise, the game continues from the new position $(\gamma[z/w], v')$.
  \end{enumerate}
  Initially, neither the cops nor the robber are placed on the graph.
  The robber wins if the cops do not win after $q$ rounds. 
  The game variant $\CR^{(k_1,k_2)}$ is defined in the same way with the modification that the robber wins if the cops never win a round.

  The \emph{node searching game} $\NS^{(k_1,k_2)}(G)$ is defined as the variant of $\CR^{(k_1,k_2)}$ in which the robber is \emph{invisible} to the cops.
  That is, the choice of the assignment $\gamma_i$ can only depend on $\gamma_{i-1}$, but not on $v$. 
  Here the cops are called \emph{searchers} instead and the robber is called \emph{fugitive}.
  In this formulation the only difference to $\CR^{(k_1,k_2)}$ is that searchers do not know the position $v$, but the fugitive knows the position~$\gamma$.

  The strategy of the cops or searchers is called \emph{monotone} if in each round it holds that $C^\gamma_v \supseteq C^{\gamma[z / \bot]}_v$.
  We say that a game variant is \emph{monotone} if the existence of a winning strategy implies the existence of a monotone winning strategy. 
\end{definition}

A winning strategy of the searchers in the game $\NS^{(k_1,k_2)}(G)$ can be specified as a sequence of positions $\gamma_1,\dots,\gamma_m \colon [x_{k_1}, y_{k_2}] \rightharpoonup V(G)$ 
while in the game $\CR^{(k_1,k_2)}$ the cop strategy depends on the moves of the robber and therefore each position $\gamma_i$ additionally depends on the robber position $v_i$ in round $i$. 
The proof of \cite[Theorem 15]{Rassmann2025} hinges on this fact and shows that in the game $\CR^{(k_1,k_2)}$ the use of non-reusable cops is restricted to a pattern involving
arbitrarily long sequences of reusing all reusable cops before utilizing a new non-reusable cop. 
We show that the situation is entirely different for $\NS^{(k_1,k_2)}$ due to the invisibility of the fugitive. 

\begin{proposition} \label{prop:search:non-reusable:first}
  The searchers have a winning strategy in $\NS^{(k_1,k_2)}(G)$ if and only if they have a winning strategy in $\NS^{(k_1,0)}(G)$ with $k_2$ initial fixed placements of non-reusable searchers. 
\end{proposition}
\begin{proof}
  Let $G$ be a graph such that the searchers have a winning strategy $\gamma_1,\dotsc,\gamma_m \colon [x_{k_1}, y_{k_2}] \rightharpoonup V(G)$ for the game $\NS^{(k_1,k_2)}(G)$ (w.l.o.g using all searchers). 
  The set $I \coloneqq  \im(\gamma_m|_{[y_{k_2}]})$ is the set of vertices in $G$ which are covered by non-reusable searchers in the course of the game. 
  We argue that placing all non-reusable searchers on $I = \{v_1,\dotsc,v_{k_2}\}$, followed by the strategy $\gamma'_{i+k_2} \coloneqq \gamma_i[y_1,\dotsc,y_{k_2}/v_1,\dotsc,v_{k_2}]$ for $i \in [m]$ induces a winning strategy for the searchers.
  Let $z_i$ be the searcher used to get from $\gamma_i$ to $\gamma_{i+1}$ and accordingly $z'_i$ from $\gamma'_i$ to $\gamma'_{i+1}$.
  Note that $\gamma'_1,\dots,\gamma'_{k_2+m}$ is not a valid strategy since it actually includes picking up a searcher $y_j$ and placing it back on the same vertex, 
  but clearly each such move can be omitted to obtain a valid strategy. 
  The set of potential positions of the fugitive after round $i\in [m]$ is given by $R'_i \coloneqq \bigcup_{v \in R_{i-1}} C^{\gamma'_i[z'_i/\bot]}_v$ with $R'_0 \coloneqq V(G)$.
  Since $\gamma_1,\dotsc,\gamma_m$ is a winning strategy for the game $\NS^{(k_1,k_2)}(G)$, for $R_i \coloneqq \bigcup_{v \in R_{i-1}} C^{\gamma_i[z_i/\bot]}_v$, $R_0 \coloneqq V(G)$ we have $R_{m+1} \coloneqq \bigcup_{v \in R_{m}} C^{\gamma_m}_v = \emptyset$.
  We show that for $i \in [m]$ we have $R'_{i+k_2} \subseteq R_i$ by induction on $i$. For the base case we have $R'_{k_2} \subseteq V(G) = R_0$ and for the inductive step 
  $$ R'_{i+k_2} = \bigcup_{v \in R'_{i+k_2-1}} C^{\gamma'_{i+k_2}[z'_{i+k_2}/\bot]}_v \subseteq \bigcup_{v \in R_{i-1}} C^{\gamma'_{i+k_2}[z'_{i+k_2}/\bot]}_v \subseteq \bigcup_{v \in R_{i-1}} C^{\gamma_{i}[z_{i}/\bot]}_v = R_i.$$
  The first inclusion follows from the inductive hypothesis and the second inclusion from $C^{\gamma'_{i+k_2}[z'_{i+k_2}/\bot]}_v \subseteq C^{\gamma_{i}[z_{i}/\bot]}_v$ for every $v \in R_{i-1}$. 
  Finally, it follows that for the final potential positions of the fugitive we have $R'_{m+k_2} \coloneqq \bigcup_{v \in R'_{m-1+k_2}} C^{\gamma'_{m-1+k_2}}_v = \emptyset$.
  After removing all duplicates from $\gamma'_1,\dots,\gamma'_{k_2+m}$ we obtain a valid winning strategy of length $m$ for the searchers. 
\end{proof}

To prove the characterization of a graph parameter by a pursuit-evasion game an important step often is to establish that the respective game is monotone. 
Alongside this, in some cases proving the monotonicity of a game is the crucial step in establishing homomorphism distinguishing closedness of the associated graph class, 
see \cite{Adler2024, Fluck2024} for a discussion. Our goal is to show monotonicity of the games introduced here in order to utilize this property to 
prove game-theoretic characterizations of $\mathcal{P}^{(k_1,k_2)}$ and $\mathcal{T}^{(k_1,k_2)}_r$ towards homomorphism distinguishing closedness.

The monotonicity of $\NS^{(k_1,0)}$ was proven in \cite{Kirousis1986} and recently \cite{Adler2024} established the monotonicity of $\CR^{(k_1,0)}_q(G)$.
We build on top of these results to show that reusability is compatible with monotonicity by replacing parts of winning strategies by monotone strategies. 
The notion of monotonicity we use here is usually referred to as \emph{robber-monotonicity} in the literature.

\begin{proposition} \label{prop:non-reusable:monotone}
    Let $k_1+k_2,q \in \mathbb{N}_+$ and $G$ be a graph. Then both games $\NS^{(k_1,k_2)}(G)$ and $\CR^{(k_1,k_2)}_q(G)$ are monotone. 
\end{proposition} 
\begin{proof}
  For the game $\NS^{(k_1,k_2)}(G)$ by \Cref{prop:search:non-reusable:first} every winning strategy of the searchers can be transformed into a strategy with 
  $k_2$ initial placements of non-reusable searchers on $I \subseteq V(G)$ followed by only utilizing reusable searchers afterwards. 
  The initial $k_2$ moves of this strategy are monotone by definition.
  For the remaining game $\NS^{(k_1,0)}(G-I)$ the searchers have a winning strategy which can be made monotone by \cite{Kirousis1986}.

  For the game $\CR^{(k_1,k_2)}_q(G)$ we argue by induction on $k_2$.
  The base case for $k_2=0$ is the main result of~\cite{Adler2024}.
  For the inductive step, let the cops have a fixed winning strategy for $\CR^{(k_1,k_2)}_q(G)$ and assume that the game variant $\CR^{(k_1,k_2-1)}_q$ is monotone for all $k_1,q \in \mathbb{N}_+$. 
  Let $P$ be the set of all positions $(\gamma_i, v_i)$ occurring in the winning strategy such that $\gamma_i$ uses a non-reusable cop for the first time.
  The index $i$ again denotes the round of the game.   
  Then for each $(\gamma_i, v_i) \in P$ the cops have a winning strategy for the game $\CR^{(k_1,k_2-1)}_{q-i}(G[C^{\gamma_i}_{v_i}] - \im \gamma_i)$ by assumption. 
  By the inductive hypothesis the cops also have a monotone strategy for the same game. 
  Now it suffices to show that the cops can force a position $(\gamma_i, v_i) \in P$ in $i$ rounds using a monotone strategy. 
  We let $P'$ be the set of all positions $(\gamma_{i-1}, v_{i-1})$ in the winning strategy preceding a position in $P$. 
  Note that all positions in $P'$ are reached by the cops only employing reusable cops. 
  Let $R(\gamma_{i-1}, v_{i-1})$ be the set of all vertices in $G$ that are occupied by reusable cops before reaching $(\gamma_{i-1}, v_{i-1}) \in P'$. 
  We define the graph $G[R(\gamma_{i-1}, v_{i-1})] + K_{k_1}$ as the induced subgraph $G[R(\gamma_{i-1}, v_{i-1})]$ with 
  a the complete graph $K_{k_1}$ attached to the position of $k_1-1$ cops by adding a perfect matching between $\{\gamma_{i-1}(x_j) : j \in [k_1 -1]\}$ and $K_{k_1}$. 
  Then the cops have a winning strategy for the game $\CR^{(k_1,0)}_{i-1+k_1}(G[R(\gamma_{i-1}, v_{i-1})] + K_{k_1})$.
  Again using \cite{Adler2024}, there exists a monotone  winning strategy for this game.
  This in turn gives a monotone strategy for the game $\CR^{(k_1,0)}(G)$ in which a position $(\gamma_{i-1}, v_{i-1}) \in P'$ can be forced in $i-1$ moves while only employing reusable cops.  
  Finally, we obtain a monotone winning strategy for $\CR^{(k_1,k_2)}_q(G)$:
  First, use the monotone strategy to force $(\gamma_{i-1}, v_{i-1})$ in $\CR^{(k_1,0)}_{i-1}(G)$, 
  followed by placing a non reusable cop (which is monotone by definition) and then use the monotone strategy for the remaining game $\CR^{(k_1,k_2-1)}_{q-i}(G[C^{\gamma_i}_{v_i}] - \im \gamma_i)$. 
\end{proof}

To characterize homomorphism indistinguishability over the class $\mathcal{T}^{(k,0)}_q$, in \cite{Fluck2024} the notion of \emph{construction trees} was introduced, building on techniques from \cite{Dvorak2010}. 
We extend this concept to accommodate the classes $\mathcal{T}^{(k_1,k_2)}_q$ and $\mathcal{P}^{(k_1,k_2)}$. 
This generalized definition forms the basis of our framework for analyzing homomorphism indistinguishability over these broader classes.

\begin{definition}
    Let $G$ be a $[x_{k_1}, y_{k_2}]$-labeled graph and $k_1+k_2 \in \mathbb{N}_+$.
    A \emph{$(k_1,k_2)$-construction tree} for $G$ is a tuple $(T, \lambda, r)$, where $(T,r)$ is a rooted tree and $\lambda \colon V(T) \to \mathcal{G}_{[x_{k_1}, y_{k_2}]}$
    is a function assigning $[x_{k_1}, y_{k_2}]$-labeled graphs to the nodes of $T$ such that
    \begin{itemize}
      \item $\lambda(r) = G$,
      \item all leaves $\ell \in V(T)$ are assigned fully labeled graphs,
      \item all internal nodes $t \in V(T)$ with exactly one child $t'$ are \emph{elimination nodes}, that is~$\lambda(t)$ is obtained from $\lambda(t')$ by deleting a label,
      \item all internal nodes $t \in V(T)$ with more than one child are \emph{product node}s, that is~$\lambda(t)$ is the product of its children,
      \item if $t \in V(T)$ is an elimination node deleting a label $y \in [y_{k_2}]$, then for all $s \prec_T t$ it holds that $y \notin L_{\lambda(s)}$.
    \end{itemize}
    The \emph{elimination depth} of $(T, \lambda, r)$ is the maximum number of elimination nodes on any path from the root $r$ to a leaf.
    If $G$ has a $(k_1,k_2)$-construction tree of elimination depth $q$ we say that $G$ is \emph{$(k_1,k_2,q)$-constructible}.
    We write $\mathcal{LT}^{(k_1,k_2)}_q$ for the class of $(k_1,k_2,q)$-constructible labeled graphs.
    If $G$ has a $(k_1,k_2)$-construction tree $T$ such that each product node $v \in V(T)$ has at most one child which is not a leaf, we say that 
    $T$ is a \emph{construction caterpillar} and $G$ is \emph{linearly $(k_1,k_2)$-constructible}.
    We write $\mathcal{LP}^{(k_1,k_2)}$ for the class of linearly $(k_1,k_2)$-constructible labeled graphs.
\end{definition}

We obtain the following characterizations for the classes $\mathcal{T}^{(k_1,k_2)}_q$ and $\mathcal{P}^{(k_1,k_2)}$, showing that all previous definitions are equivalent for unlabeled graphs. 

\begin{theorem} \label{thm:path:characterization}
  For a graph $G$ and $k_1, k_2 \in \mathbb{N}$ with $k_1+k_2\geq1$ the following are equivalent:
  \begin{enumerate}
    \item\label{path:decomposition} $G \in \mathcal{P}^{(k_1,k_2)}$, i.e., $G$ has a path decomposition of width $(k_1-1,k_2)$.
    \item\label{path:cover}         $G$ admits a linear $(k_1, k_2)$-pebble forest cover.
    \item\label{path:game}          The searchers have a winning strategy for the game $\NS^{(k_1,k_2)}(G)$.
    \item\label{path:construction}  $G \in \mathcal{LP}^{(k_1,k_2)}$, i.e., $G$ is linearly $(k_1,k_2)$-constructible.
  \end{enumerate}
\end{theorem} 
\begin{proof}
    $\ref{path:construction} \Longleftrightarrow \ref{path:decomposition} :$

    Let $G$ be linearly $(k_1,k_2)$-constructible and let $(T, \lambda, r)$ be a (without loss of generality) binary $(k_1,k_2)$-construction caterpillar.
    The caterpillar $T$ contains a central path $P$ from $r$ to a leaf with maximum distance from $r$.
    For each node $p \in V(P)$ we define $\beta(p)$ to be the set of all labeled vertices in the graph $\lambda(p) \in \mathcal{G}_{[x_{k_1}, y_{k_2}]}$.
    We show that $(P, \beta, r)$ is a path decomposition of $G$ of width $(k_1,k_2)$.
    For the endpoint $\ell \in V(P)$ opposite from $r$ we define $S_{\ell}$ as the set of $[y_{k_2}]$-labeled vertices in $\lambda(\ell)$.
    Then $|S_{\ell}| \leq k_2$ and for every $p \preceq_P \ell$ it holds that $|\beta(p) \setminus S_{\ell}| \leq k_1$ since each label in $[y_{k_2}]$ is assigned to at most one vertex in $\bigcup_{p \preceq_T \ell}\beta(p)$.
    It remains to show that $(P, \beta, r)$ is a path decomposition.
    Since every leaf of the construction caterpillar must be fully labeled, for each edge $uv \in E(G)$ there exists a leaf $\ell \in V(T)$ such that $u,v \in V(\lambda(\ell))$.
    If $\ell$ is a forget node, then $\ell \in V(P)$ and otherwise $\ell$ is the child of a join node $p \in V(P)$ so $u,v \in V(\lambda(p)) = \beta(p)$.
    Also, $\beta^{-1}(v)$ is the set of nodes $p \in V(P)$ such that $p$ is labeled in $\lambda(p)$. Since labels can only be introduced at leaves and only labeled vertices are identified at product nodes, this set is connected in $P$.

    For the converse direction let $(P,r,\beta)$ be a nice path decomposition of width $(k_1-1,k_2)$ of $G$.
    The $(k_1,k_2)$-construction caterpillar $(T,r,\lambda)$ of $G$ is defined as follows:
    We obtain $T$ from $P$ by appending a new leaf $\ell_p$ at every introduce node of $p \in V(P)$ and extend $\beta$ to $T$ by setting $\beta(\ell_p) = \beta(p)$ for every introduce node $p \in V(P)$.
    For the assignment of labeled graphs let $t \in V(T)$ and set $\lambda(t) \coloneqq G[\bigcup_{t \preceq_{T} s} \beta(s)]$.
    For the labeling of $\lambda(t)$ construct a coloring $c \colon V(G) \to [x_{k_1}, y_{k_2}]$ by traversing the caterpillar $T$ from the root to the leaves. Whenever $t \in V(T)$ is a forget node with child $s$
    and $\beta(s) \setminus \beta(t) = \{v\}$ if $v \in S$, we set $c(v)$ to be the smallest (with respect to its index) element in $[y_{k_2}]$ that is not used in $\beta(t)$ and otherwise to be the smallest element in $[x_{k_1}]$ that is not used in $\beta(t)$.
    Now for each node $t \in V(T)$ the labeling of $\lambda(t)$ is given by ${c|_{\beta(t)}}^{-1}$.
    Then $(T, \lambda, r)$ is a $(k_1,k_2)$-construction caterpillar of $G$:
    We have $\lambda(r) = G[\bigcup_{r \preceq_{T'} s} \beta(s)] = G$.
    All leaves are assigned fully labeled graphs (empty graphs are fully labeled).
    Let $t \in V(T')$ be an internal node with exactly one child, then $t$ was a forget node in the nice path decomposition and now is an elimination node in the construction tree.
    Otherwise, $t$ was an introduce node and now has children $\ell_t$ and $s$ with $\lambda(\ell_t) = G[\beta'(t)]$ being fully labeled and $\lambda(s) = G[(\beta'(t) \setminus \{v\}) \cup \bigcup_{s \preceq_{T'} s'} \beta(s')]$.
    In particular, every product node in $T$ has most one child which is not a leaf.
    As the labelings of $\lambda(\ell_t)$ and $\lambda(s)$ coincide on $\beta'(t) \setminus \{v\}$, we obtain $\lambda(t) = \lambda(\ell_t) \lambda(s)$.
    By \Cref{lem:nice:decomp} after deleting a label $y \in [y_{k_2}]$ the same label does not occur at any node above.

    $\ref{path:decomposition} \Longleftrightarrow \ref{path:cover} :$

    Let $(P, \beta)$ be a nice path decomposition of width $(k_1-1,k_2)$ with exception set $S$ and let $P=p_1 \dotsc p_n$ inducing the linear ordering $p_1 \prec_P p_2 \prec_P \ldots \prec_P p_n$.
    We define $\tau \colon V(G) \to V(P), v \mapsto p_{\min \{i \in [n] : v \in \beta(p_i)\}}$, so for $v \in V(G)$ the node $\tau(v) \in P$ has smallest possible distance to the endpoint $p_1$.
    The function $\tau$ is injective since the path decomposition is nice.
    In particular, by~\Cref{lem:nice:decomp} the exception set $S$ is linearly ordered by setting $s_i \prec_S s_j$ exactly if $\tau(s_i) \prec_P \tau(s_j)$.
    Let $S = \{s_1 \ldots, s_{k_2}\}$ with $s_1 \prec_S \ldots \prec_S s_{k_2}$.
    We define $V(L) \coloneqq V(G)$ and set $u \prec_L v$ exactly if $\tau(u) \prec_P \tau(v)$ and $u,v$ are in the same connected component of $G$.
    Since $\preceq_P$ is a linear order and $\tau$ is injective also each component of $\preceq_L$ is a linear order and hence $L$ is a linear forest.
    Let the vector $\overline{r}$ contain the minimal elements of the order $\preceq_L$.
    The pebbling function $p \colon V(G) \to [x_{k_1}, y_{k_2}]$ is defined by induction on $\preceq_L$.
    Consider $v \in V(L)$ and assume that $p(v')$ is defined for all $v' \prec_L v$, then set
    $$ p(v) \coloneqq \begin{cases}
            \min [x_{k_1}] \setminus \{p(w) : w \in \beta(\tau(v)) \setminus \{v\}\} & \text{ if } v \notin S \\
             y_j & \text{ if } v = s_j \in S
        \end{cases} $$
    We argue that $(L, \overline{r}, p)$ is a linear $(k_1,k_2)$-pebble forest cover.
    For $uv \in E(G)$ it holds $\tau(u) \preceq_P \tau(v)$ or $\tau(v) \preceq_P \tau(u)$ since $\preceq_P$ is a linear order and $u,v$ are in the same connected component of $G$.
    Now let $uv \in E(G), w \in V(L)$ with $u \prec_L v$ and $u \prec_L w \preceq_L v$.
    Then $w$ must be in the connected component of $u,v$ and there exists a node $p_i \in V(P)$ with $u,v \in \beta(p_i)$ such that $\tau(u) \prec_P \tau(w) \preceq_P \tau(v) \preceq_P p_i$.
    Thus, $u \in \beta(\tau(v)) \cap \beta(p_i)$ and since $\beta^{-1}(u)$ is connected this yields $u \in \beta(\tau(w))$.
    By the definition of the pebbling function this implies $p(u) \neq p(w)$.
    Finally, let $u \in V(G)$ with $p(u) \in [y_{k_2}]$, so $u \in S$ by the definition of $p$.
    For $w \in V(G) \setminus \{u\}$ with $p(w) \in [x_{k_1}]$ this directly implies $p(u) \neq p(w)$. 
    If $p(w) \in [y_{k_2}]$ we have $u \neq w = s_i \in S$ and hence $p(u) \neq p(w)$ by the linear ordering of $S$. 

    Let $(F, \overline{r}, p)$ be a linear $(k_1,k_2)$-pebble forest cover of $G$. We define a path decomposition $(P, \beta)$ of $G$ as follows:
    First, assume that $G$ is connected. Then the forest $F$ consist of a single path with one distinguished endpoint $r$:
    Suppose that $F$ contains two connected components, since $G$ is connected there must be an edge between them and since $F$ is a forest cover, the endpoints must be in the same linear order. 
    We set $P \coloneqq F$ and for $v \in V(P)$ let 
    $$\beta(v) \coloneqq \{u \preceq_P v : \text{ for all } w \in V(G) \text{ it holds } u \prec_P w \preceq_P v \Rightarrow p(u) \neq p(w)\}.$$
    Then $\bigcup_{v \in V(P)} \beta(v) = V(G)$ and for each edge $uv \in E(G)$ it holds $u,v \in \beta(v)$ or $u,v \in \beta(u)$ and
    for every $v \in V(G)$ the set $\beta^{-1}(\{v\})$ precedes $v$ in $P$ and therefore is connected.
    We let $\ell$ be the endpoint of $P$ opposite from $r$ and set $S \coloneqq \{w \in V(P): w \preceq_P \ell,~ p(w) \in [y_{k_2}]\}$.
    Then $|S| \leq k_2$ and $\max_{t \preceq_P \ell} |\beta(t) \setminus S| \leq k_1$.
    Now let $G$ have connected components $C_1,\dotsc,C_m$. By the previous construction, we obtain path decompositions $(P_1,\beta_1),\dotsc,(P_k,\beta_m)$ of the connected components each of width $(k_1-1, k_{2,m})$.
    Then $(P, \beta)$ with $P \coloneqq P_1\dotsc P_m$ and $\beta(v) \coloneqq \beta_i(v)$ for the unique index $i \in [m]$ with $v \in \beta_i(P-i)$ is a path decomposition for $G$ of width $(k_1-1, k_{2,m})$:
    We have $\bigcup_{v \in V(P)} \beta(v) = \bigcup_{i=1}^k \bigcup_{v \in V(P_i)} \beta_i(v) = \bigcup_{i=1}^k V(C_i) = V(G)$,
    each edge $uv \in E(G)$ is contained in some component $C_i$ and hence in a bag of $(P_i, \beta_i)$, and
    for all $v \in V(G)$ the set $\beta^{-1}(v) = \beta^{-1}_i(v)$ is connected in $P_i$ and thus in $P$.
    Moreover, for each decomposition $(P_i, \beta_i)$ there exists a set $S_i$ of exceptions from which we define $S \coloneqq \bigcup_{i=1}^m S_i$ with $|S| = \sum_{i=1}^m k_{2,i} \leq k_2$.
    Finally, $\max_{t \preceq_P \ell_m} |\beta(t) \setminus S| = \max_{i \in [m]}\max_{t \preceq_{p_i} \ell_i} |\beta(t) \setminus S_i| \leq k_1$.

    $\ref{path:decomposition} \Longleftrightarrow \ref{path:game} :$

    Let $(P,\beta)$ be a nice path decomposition for $G$ of width $(k_1-1, k_2)$.
    The winning strategy for the searchers in $\NS^{(k_1,k_2)}(G)$ is as follows:
    In the first $k_2$ rounds of the game place the non-reusable searchers successively on the elements of $S$.
    We let $P=p_1 \dotsc p_n$ and for $i \in [n]$ the placement of the searchers in round $k_2+i$ be $S \cup \beta(p_i)$.
    Since the path decomposition is nice this yields a (re)placement of exactly one searcher in every round
    and the number of reusable searchers placed in round $i$ is bounded by $|\beta(p_i)| \leq k_1$.
    For every edge $p_ip_{i+1} \in E(P)$ the searcher-occupied set $S \cup (\beta(p_i) \cap \beta(p_{i+1}))$ separates $\cup_{j=1}^i \beta(p_j)$ from $\cup_{j=i+1}^n \beta(p_j)$.
    Thus, in each round the fugitive is forced to play in $\cup_{j=i+1}^n \beta(p_j)$ and caught eventually since every edge $uv \in E(G)$ is contained in some bag $\beta(p_j)$.

    We turn a winning strategy for the searchers in $\NS^{(k_1,k_2)}(G)$ into a path decomposition $(P, \beta)$ of $G$ with width $(k_1-1, k_2)$:
    By \Cref{prop:search:non-reusable:first} we may assume that in the winning strategy, the searchers first place the $k_2$ non-reusable searchers 
    on a set $I \subseteq V(G)$ with $|I|\leq k_2$ and then have a winning strategy for the game $\NS^{(k_1,0)}(G-I)$.
    By \cite[5.1]{Bienstock1991} there is a path decomposition $P$ of width $(k_1-1,0)$ for the graph $G-I$. 
    We define the set of exceptions $S$ to be the same as $I$ and infer that $P$ with $I$ added to each bag is a path decomposition of width $(k_1-1,k_2)$ for $G$. 
\end{proof}

\begin{theorem} \label{thm:tree:characterization}
  For a graph $G$ and $k_1,k_2,q \in \mathbb{N}$ with $k_1+k_2,q \geq 1$ the following are equivalent:
  \begin{enumerate}
    \item\label{tree:decomposition}  $G\in \mathcal{T}^{(k_1,k_2)}_q$, i.e., $G$ has a tree decomposition of width $(k_1-1,k_2)$ and depth~$q$.
    \item\label{tree:cover}          $G$ admits a $(k_1,k_2)$-pebble forest cover of depth $q$.
    \item\label{tree:game}           The cops have a winning strategy for the game $\CR^{(k_1,k_2)}_q(G)$.
    \item\label{tree:construction}   $G \in \mathcal{LT}^{(k_1,k_2)}_q$, i.e., $G$ is $(k_1,k_2,q)$-constructible.
  \end{enumerate}
\end{theorem}
\begin{proof}
  We follow the proofs of \cite[Theorem 14, Lemma 16]{Fluck2023} and show how their constructions preserve reusability constraints:

    $\ref{tree:construction} \Longleftrightarrow \ref{tree:decomposition}:$

    Let $G$ be $(k_1,k_2,q)$-constructible and let $(T, \lambda, r)$ be a $(k_1,k_2)$-construction tree of elimination depth $q$.
    For each node $t \in V(T)$ we define $\beta(t)$ to be the set of all labeled vertices in the graph $\lambda(t) \in \mathcal{G}_{[x_{k_1}, y_{k_2}]}$.
    We show that $(T, \beta, r)$ is a tree decomposition of $G$ of width $(k_1,k_2)$ and depth $q$.
    Concerning the depth of the decomposition, the number of vertices contained in bags on a path from $r$ to a leaf in $T$ is exactly the number of elimination nodes on the path which is $q$.
    For each leaf $\ell \in V(T)$ we define $S_{\ell}$ as the set of $[y_{k_2}]$-labeled vertices in $\lambda(\ell)$.
    For the width we have $|S_{\ell}| \leq k_2$ and for every $t \preceq_T \ell$ it holds that $|\beta(t) \setminus S_{\ell}| \leq k_1$ since each label in $[y_{k_2}]$ is assigned to at most one vertex in $\bigcup_{t \preceq_T \ell}\beta(t)$.
    Finally, $(T, \beta, r)$ is a tree decomposition by \cite[Appendix A.1]{Fluck2023}.

    For the converse direction let $(T,r,\beta)$ be a nice tree decomposition of width $(k_1,k_2)$ and depth $q$.
    We construct a new tree decomposition $(T', r, \beta')$ for which at every introduce node $t \in V(T)$ we append a new leaf $\ell_t$ to obtain $T'$.
    For every $t \in V(T')$ we now set $\lambda(t) \coloneqq G[\bigcup_{t \preceq_{T'} s} \beta(s)]$, which is the induced subgraph of $G$ on all vertices that occur in bags $\beta(s)$ for $t \preceq_{T'} s$.
    We define the bag $\beta'(t) = \beta(t)$ for $t \in V(T)$ and $\beta'(\ell_t) = \beta(t)$ for the new leaves.
    The explicit labeling is then given by the inverse of the injective function $c \colon V(G) \to [x_{k_1}, y_{k_2}]$ that assigns labels as follows:
    We traverse the tree $T'$ from the root to the leaves. Whenever $t$ is a forget node with child $s$ and $\beta'(s) \setminus \beta'(t) = \{v\}$ if $v \in S_\ell$ for some $t \preceq_{T'} \ell$, we set $c(v)$ to be the smallest element in $[y_{k_2}]$ that is not used in $\beta'(t)$
    and otherwise to be the smallest element in $[x_{k_2}]$ that is not used in $\beta'(t)$.
    Now for each $t \in V(T')$ the labeling of $\lambda(t)$ is given by ${c|_{\beta(t)}}^{-1}$.
    Then $(T', \lambda, r)$ is a $(k_1,k_2)$-construction tree of $G$ of depth $q$:
    By \Cref{lem:nice:decomp} after deleting a label $y \in [y_{k_2}]$ it does not occur at at any node above.
    We have $\lambda(r) = G[\bigcup_{r \preceq_{T'} s} \beta(s)] = G$.
    All leaves are assigned fully labeled graphs (empty graphs are fully labeled).
    Let $t \in V(T')$ be an internal node with exactly one child, then $t$ was a forget node in the nice tree decomposition and now is an elimination node in the construction tree.
    If $t \in V(T')$ was a join node in the nice tree decomposition with two children $s_1,s_2$ and then $\lambda(t) = \lambda(s_1)\lambda(s_2)$.
    Otherwise, $t$ was an introduce node and now has children $\ell_t$ and $s$ with $\lambda(\ell_t) = G[\beta'(t)]$ being fully labeled and $\lambda(s) = G[(\beta'(t) \setminus \{v\}) \cup \bigcup_{s \preceq_{T'} s'} \beta(s')]$.
    As the labelings of $\lambda(\ell_t)$ and $\lambda(s)$ coincide on $\beta'(t) \setminus \{v\}$, we obtain $\lambda(t) = \lambda(\ell_t) \lambda(s)$.

    $\ref{tree:decomposition} \Longleftrightarrow \ref{tree:cover}:$

    Let $(T, r, \beta)$ be a rooted tree decomposition of $G$ of width $(k_1-1,k_2)$ and depth $q$. Again assume that this decomposition is nice.
    For each $v \in V(G)$ let $\tau(v)$ be the unique node in $T$ such that $v \in \beta(\tau(v))$ and $\tau(v)$ has the smallest possible distance to the root.
    Then $\tau \colon V(G) \to V(T)$ is well-defined and injective since the decomposition is nice.
    The forest cover $(F, \overline{r})$ of $G$ is then defined as follows:
    We define $V(F) \coloneqq V(G)$ and set $u \preceq_F v$ exactly if $\tau(u) \preceq_T \tau(v)$, so $uv \in E(F)$ exactly if $\tau(u)\tau(v) \in E(T)$.
    Let the vector $\overline{r}$ contain exactly the minimal elements of the order $\preceq_F$.
    The pebbling function $p \colon V(G) \to [x_{k_1}, y_{k_2}]$ is defined by induction on $\preceq_F$.
    Assume that $p(v')$ is defined for all $v' \prec_F v$, then set
    $$ p(v) \coloneqq \begin{cases}
            \min [x_{k_1}] \setminus \{p(v') : v' \in \beta(\tau(v)) \setminus \{v\}\} & \text{ if } v \notin S_{\ell} \text{ for all leaves } \ell \text{ with } v \prec_F \ell \\
            \min [y_{k_2}] \setminus \{p(v') : v' \in \beta(\tau(v)) \setminus \{v\}\} & \text{ if } v \in S_{\ell} \text{ for some leaf } \ell \text{ with } v \prec_F \ell
        \end{cases} $$
    We argue that $(F, \overline{r}, p)$ is a $(k_1,k_2)$-pebble forest cover of depth $q$.
    If $uv \in E(G)$ there exists $t \in V(T)$ with $u,v \in \beta(t)$. Then it holds that $\tau(u), \tau(v) \preceq_T t$ and because $T$ is a tree we have $\tau(u) \preceq_T \tau(v)$ or $\tau(v) \preceq_T \tau(u)$.
    This yields $u \preceq_F v$ or $v \preceq_F u$.
    Now let $uv \in E(G)$ with $u \prec_F v$ and $u \prec_F w \preceq_F v$.
    Then there exists a node $t \in V(T)$ with $u,v \in \beta(t)$ such that $\tau(u) \prec_T \tau(w) \preceq_T \tau(v) \preceq_T t$.
    Thus, $u \in \beta(\tau(v)) \cap \beta(t)$ and since $\beta^{-1}(u)$ is connected this yields $u \in \beta(\tau(w))$.
    By the definition of the pebbling function this yields $p(u) \neq p(w)$.
    Finally, let $u \in V(G)$ with $p(u) \in [y_{k_2}]$ and $u \prec_F w$.
    Then $u \in S_{\ell}$ for some leaf $\ell$ with $u \prec_F \ell$ and by monotonicity $u \in \beta(\tau(u)) \cap S_\ell \subseteq \beta(\tau(w))$
    and thus $u \in \beta(\tau(w)) \setminus \{w\}$ which yields $p(u) \neq p(w)$.

    Let $(F, \overline{r}, p)$ be a $(k_1,k_2)$-pebble forest of depth $q$ of $G$.
    The tree decomposition $(T, r', \beta)$ is defined as follows:
    the root $r'$ is a new node that is connected to all roots in $\overline{r}$ and the corresponding trees.
    We set $\beta(r') \coloneqq \emptyset$ and for every $t \in V(F)$ set $\beta(t) \coloneqq \{u \preceq_F t : \text{ for all } w \in V(G), u \prec_F w \preceq_F t \Rightarrow p(u) \neq p(w)\}$.
    Then $\bigcup_{t \in V(T)} \beta(t) = V(G)$, for each edge $uv \in E(G)$ it holds $u,v \in \beta(v)$ or $u,v \in \beta(u)$ and
    for every $v \in V(G)$ the set $\beta^{-1}(\{v\})$ is above $v$ in $F$ and hence connected.
    For each leaf $\ell \in V(T)$ we set $S_{\ell} \coloneqq \{w \in V(T): w \preceq_T \ell,~ p(w) \in [y_{k_2}]\}$.
    Then $|S_{\ell}| \leq k_2$ and $\max_{t \preceq_T \ell} |\beta(t) \setminus S_{\ell}| \leq k_1$.

    $\ref{tree:cover} \Longrightarrow \ref{tree:game}$:

    Let $(F, \overline{r}, p)$ be a $(k_1,k_2)$-pebble forest cover of depth $q$ of $G$.
    We construct a monotone winning strategy for the cops for the game $\CR^{(k_1,k_2)}_q(G)$ such that for every position $(\gamma, v)$ it holds that
    \begin{enumerate}
        \item\label{invariant:prec} $\im(\gamma) \preceq_F v$
        \item\label{invariant:image} $\im(\gamma) = \{u \preceq_F \max \im(\gamma) : \forall w \in V(G) ~ u \prec_F w \preceq_F \max \im(\gamma) \Rightarrow p(u) \neq p(w)\}$
    \end{enumerate}
    Let $(\gamma, v)$ be the position at the beginning of a round.
    We let $w$ be the minimal vertex with $\im(\gamma) \prec_F w \preceq v$ and if $p(w) \notin [y_{k_2}]$ let $\gamma(x_i) \in \im(\gamma)$ such that $p(w) = p(\gamma(x_i))$.
    Then the strategy of the cops is to choose the new position
    $$ \gamma' \coloneqq \begin{cases}
            \gamma[p(w) / w] & \text{ if } p(w) \in [y_{k_2}] \\
            \gamma[x_i / w]  & \text{ otherwise }
        \end{cases} $$
    That is, the cop $z \in \{p(w), x_i\}$ is picked up, then the robber moves to a vertex $v'$ in $G - \im(\gamma[z/w])$ and the new position is given by $(\gamma', v')$.
    Note that if $p(w) \in [y_{k_2}]$, then no ancestor of $w$ holds the same pebble and hence the cop $p(w)$ has not been used so far.
    If $\im(\gamma) = \emptyset$ the vertex $w$ is the root of some tree in $F$ and otherwise the child of $\max \im(\gamma)$.
    Thus, condition \ref{invariant:image} is also satisfied for $\gamma'$.
    Now assume that there is a path from $w$ to $v'$ in $G[\im(\gamma(z / \bot))]$ and $w \notpreceq_F v'$.
    On this path there must be a vertex that comes before $w$ in the order $\preceq_F$ as $(F, \overline{r})$ is a forest cover of $G$.
    Thus, there exists an edge $u_1u_2 \in E(G)$ on the path from $w$ to $e'$ in $G[\im(\gamma(z / \bot))]$ with $u_1 \prec_F \max \im(\gamma) \prec_F w \preceq_F u_2$.
    Since $u_1 \notin \im(\gamma)$ we obtain from \ref{invariant:image} that there exists some $u_3 \in \im(\gamma)$ with $p(u_1) = p(u_3)$ and $u_1 \prec_F u_3 \prec_F w \preceq_F u_2$,
    which is a contradiction to $(F, \overline{r}, p)$ being a $(k_1,k_2)$-pebble forest cover.

    $\ref{tree:game} \Longrightarrow \ref{tree:decomposition}$:

    Assume that the cops have a monotone winning strategy for the game $\CR^{(k_1,k_2)}_q(G)$.
    From a position $(\gamma, v)$ of the game, the move of the robber only depends on the component $C^{\gamma}_v$ rather than the vertex $v$.
    We set $S(\gamma, C^{\gamma}_v)$ to be the set of cop positions in the round after the robber chose the edge $v$.
    That is, $S(\gamma, C^{\gamma}_v) = \im(\gamma')$ whenever $(\gamma', v')$ is the next position.
    W.l.o.g we assume that for all cop positions $\gamma$ and $v' \in C^{\gamma}_v$ it holds that the cops always place the next cop in $C^{\gamma}_v$ and $S(\gamma, C^{\gamma}_v) = S(\gamma, C^{\gamma}_{v'})$.
    We turn the strategy of the cops into a tree decomposition of $G$ as follows:
    \begin{enumerate}
        \item $\beta(r) \coloneqq \emptyset$
        \item Let $C_1,\dotsc,C_p$ be the connected components of $G$, then add vertices $t_1,\dotsc,t_p$ as children of $r$ to $T$ and set $\lambda(t_i) \coloneqq (\emptyset, V(C_i))$, $\beta(t_i) \coloneqq S(\emptyset, C_i)$
        \item For $s \in V(T) \setminus \{r\}$ with $\lambda(s) = (X, C)$ let $C_1, \dotsc, C_p$ be the connected components of $C \setminus S(X,C)$, then add children $t_1,\dotsc,t_p$ to $s$ with $\lambda(t_i) \coloneqq (S(X,C), V(C_i))$
              and $\beta(t_i) \coloneqq S(S(X,C), V(C_i))$
        \item For each leaf $\ell \in V(T)$ we set $S_{\ell}$ to be the set of $[y_{k_2}]$-cop positions occurring in $\bigcup_{t \preceq_T \ell} \beta(t)$
    \end{enumerate}
    Then $(T,r,\beta)$ is a tree decomposition of width $(k_1,k_2)$ and depth $q$:
    Regarding the depth of the decomposition, note that each edge corresponds to moving a single cop and hence from the empty root to any leaf at most $q$ cop positions may be introduced.
    For the width we have the following:
    Each leaf $\ell \in V(T)$ corresponds to a final position of the game and the vertices in $S_{\ell}$ are precisely the non-reusable cops that are placed in this position.
    As otherwise only reusable cops are placed in this position, it holds $\max_{t\preceq_T \ell} |\beta(t) \setminus S_{\ell}| \leq k_1$.
    That $(T,r,\beta)$ is a tree decomposition is shown in \cite[Appendix A.2]{Fluck2023}.
\end{proof}

\begin{remark} \label{rem:hierarchy}
  Regarding inclusions between classes $\mathcal{T}^{(k_1,k_2)}$ (or $\mathcal{P}^{(k_1,k_2)}$) for varying parameters $k_1$ and $k_2$ we observe that 
  the proof of \cite[Theorem 13]{Rassmann2025} can be used to show a complete classification of all inclusions, which in particular separates $\mathcal{T}^{(k_1,k_2)}$
  from $\mathcal{T}^{(k'_1,k'_2)}$ for $(k_1,k_2) \neq (k'_1,k'_2)$ (and likewise for $\mathcal{P}^{(k_1,k_2)}$, except that $\mathcal{P}^{(1,k_2)} \not\subseteq \mathcal{P}^{(0, k'_2)}$ for all $k'_2$ by \Cref{prop:search:non-reusable:first}).
\end{remark}

We observe that classes the $\mathcal{T}^{(k_1,k_2)}_q$ and $\mathcal{P}^{(k_1,k_2)}$ admit usual closure properties, except that $\mathcal{P}^{(k_1,k_2)}$ is not closed under taking disjoint unions. 

\begin{proposition} \label{prop:closure}
  Let $k_1+k_2 \in \mathbb{N}_+$.
  The class $\mathcal{T}^{(k_1,k_2)}_q$ is closed under taking disjoint unions, summands, and minors. 
  The class $\mathcal{P}^{(k_1,k_2)}$ is closed under taking summands and minors but not under taking disjoint unions for $k_2 \geq 1$. 
\end{proposition} 

The class $\mathcal{P}^{(k_1,k_2)}$ formalizes the notion of reusability for path decompositions in an appropriate sense as exemplified by the characterization through $\NS^{(k_1,k_2)}$.
However, to overcome the obstacle that this class is not closed under disjoint unions we define the class $\dot{\cup}\mathcal{P}^{(k_1,k_2)}$ as the closure of $\mathcal{P}^{(k_1,k_2)}$ under disjoint unions.
This class formalizes the notion of \emph{componentwise} restricted reusability on a graph. The next proposition, which follows directly from \Cref{thm:path:characterization}, makes this explicit.

\begin{proposition}\label{prop:path:component:characterization}
    For a graph $G$ and $k_1+k_2 \in \mathbb{N}_+$ the following are equivalent:
    \begin{enumerate}
      \item\label{path:comp:decomposition} $G\in \dot{\cup}\mathcal{P}^{(k_1,k_2)}$, i.e., $G$ has a path decomposition of component width $(k_1-1,k_2)$.
      \item\label{path:comp:cover}         $G$ admits a linear $(k_1, k_2)$-pebble component forest cover.
      \item\label{path:comp:game}          For each component $C \in \mathcal{C}_G$ the searchers have a winning strategy for $\NS^{(k_1,k_2)}(G[C])$.
    \end{enumerate}
\end{proposition} 
  \section{Homomorphism indistinguishability and logical equivalence} \label{sec:hom-ind-log}
In this section, we characterize homomorphism indistinguishability over the classes $\mathcal{P}^{(k_1,k_2)}$ and $\mathcal{T}^{(k_1,k_2)}_q$ using logics with restricted requantification. 
Our approach builds on and extends the techniques developed in \cite{Dvorak2010,Fluck2024}, establishing tight connections between the syntactic structure of formulas in these logics and the combinatorial structure of graphs in $\mathcal{LP}^{(k_1,k_2)}$ and $\mathcal{LT}^{(k_1,k_2)}_q$,
specifically through the notions of construction caterpillars and construction trees.

\subsection{Construction caterpillars and restricted conjunction}
We start by giving the definition of finite variable counting logic with restricted conjunction and requantification, extending a definition from \cite{Montacute2024}.
\begin{definition}
    We define the set of logical formulas ${\land}\mathsf{C}^{(k_1,k_2)}_{\infty \omega}$ over the variable sets $[x_{k_1}, y_{k_2}]$ and $\mathcal{W} = \{w_1,w_2,\dots\}$.
    The \emph{non-counting formulas} of the logic are given by
    $$
        \varphi \Coloneqq z_i = z_j \ \vert \ R(\overline{z}) \ \vert \ \neg p \ \vert \ \bigvee_{i \in I} \psi_i \ \vert \ \bigwedge_{j \in J} \psi_j \ \vert \ \exists z_i (z_i = w_\ell \land \psi(\overline{z}, \overline{w}))
    $$
    for $z_i,z_j \in [x_{k_1}, y_{k_2}]$, $\overline{z} \in [x_{k_1}, y_{k_2}]^n$, $p$ atomic, $I$ and $J$ countable index sets, $\bigwedge_{j \in J} \psi_j$ a restricted conjunction, and a non-counting formula $\psi(\overline{z}, \overline{w})$ with $\overline{w} \in \mathcal{W}^m$, $w_{\ell}\notin \overline{w}$.
    Here, \emph{restricted conjunction} means that at most one formula $\psi_j$ containing a quantifier is not a sentence. 
    Furthermore, the logic contains the formulas
    $$
        \varphi \Coloneqq \exists^{n} (w_{\ell_1},\dotsc,w_{\ell_m}) \psi(w_{\ell_1},\dotsc,w_{\ell_m}) \ \vert \ \psi_1 \lor \psi_2
    $$
    for $n,m \in \mathbb{N}$, $w_{\ell_1},\dotsc,w_{\ell_m} \in W$, a non-counting formula $\psi$, and $\psi_1, \psi_2 \in {\land}\mathsf{C}^{(k_1,k_2)}_{\infty \omega}$.
    We additionally require that only variables from $[x_{k_1}]$ are requantified.
    The fragment ${\land}\mathsf{C}^{(k_1,k_2)}$ is defined by additionally requiring that all conjunctions and disjunctions are finite.

    We call a non-counting formula \emph{primitive} if it contains no disjunction and every restricted conjunction does not contain a sentence.
    A formula $\varphi \in {\land}\mathsf{C}_{\infty \omega}$ is called \emph{primitive} if it is of the form $\varphi = \exists^{n} \overline{w} \ \psi(\overline{w})$ for a primitive non-counting formula $\psi$.
\end{definition}

We first prove that there is a normal form for restricted conjunction counting logic with respect to requantification and primitivity, enabling a more direct correspondence between the syntax of formulas and construction caterpillars.
The idea is to translate the scheme from \Cref{prop:search:non-reusable:first} into the language of logic: It suffices to first fix all non-requantifiable variables followed by a well-behaved employment of requantifiable variables.
\begin{proposition} \label{prop:primitive}
    Every sentence $\varphi \in {\land}\mathsf{C}^{(k_1,k_2)}_{\infty \omega}$ is logically equivalent to disjunction of sentences of the form 
    $$ \exists^n \overline{w} \exists y_1 \dotsc \exists y_{k_2} \bigwedge_{i\in[k_2]} y_i = w_{\ell_i} \land \chi $$
    for a primitive non-counting formula $\chi$ only containing quantification over variables from~$[x_{k_1}]$. 
\end{proposition}
\begin{proof}
  First note that $\varphi$ is is logically equivalent to a disjunction of primitive ${\land}\mathsf{C}^{(k_1,k_2)}_{\infty \omega}$-sentences 
  by \cite[Proposition 5.8]{Montacute2024} together with the observation that applying the rewrite rule~$\exists z \bigvee_{i \in I} \varphi_i \mapsto \bigvee_{i \in I} \exists z \varphi_i$ does not affect requantification.
  In each of these primitive sentences then every conjunctive subformula contains only literals and at most one quantified formula with a free variable.
  Thus, we can pull each quantification $\exists y_j$ and the guarding equality $y_j = w_{\ell}$ over all conjunctions since the other conjuncts do not contain the variable $y_j$.
  Otherwise, $y_j$ would occur free and bound and thus would be requantified. 
\end{proof}

In \cite{Montacute2024} it was shown that equivalence in ${\land}\mathsf{C}^{(k,0)}_{\infty \omega}$ is the same as homomorphism indistinguishability over $\mathcal{P}^{(k,0)}$.
This result was proven by evoking a categorical meta-theorem from \cite{Dawar2021}.
We utilize the constructive nature of the proofs in \cite{Dvorak2010,Fluck2024} to give a new combinatorial proof of the result that also adapts to the setting with restricted requantification.
The main technical difficulty here is that the class $\mathcal{LP}^{(k_1,k_2)}$ is not closed under taking products.
However, this aligns with the fact that ${\land}\mathsf{C}^{(k_1,k_2)}$ is also not closed under arbitrary conjunctions.

As a first step, we show that homomorphism counts from graphs in $\mathcal{LP}^{(k_1,k_2)}$ are ${\land}\mathsf{C}^{(k_1,k_2)}$-definable by inductively building up the formula along
a construction caterpillar.

\begin{lemma} \label{lem:path:hom:count:definable}
    Let $F \in \mathcal{LP}^{(k_1,k_2)}$ and $m \in \mathbb{N}$.
    There exists a formula $\varphi^F_m \in {\land}\mathsf{C}^{(k_1,k_2)}$ with $\free(\varphi^F_m) = L_F$ such that for each $[x_{k_1}, y_{k_2}]$-labeled graph $G$ with $L_F \subseteq L_G$ it holds that
    \begin{center}
        $G \models \varphi^F_m $ if and only if $\hom(F, G) = m$. 
    \end{center}
\end{lemma}
\begin{proof}
    Let $F \in \mathcal{LP}^{(k_1,k_2)}$ and a construction caterpillar $(T, \lambda,r)$ for $F$ be given.
    For a leaf $\ell \in \leaf(T)$ we let $\elim(T, \ell)$ be the sequence of labels being eliminated in $T$ on the path from $r$ to $\ell$
    and $\elim(T)$ be the sequence of labels being eliminated in $T$.
    We show by induction on the structure of $T$ that
    \begin{equation}
      \label{eq:path:hom-count}
      \tag{$\ast$}
          \hom(F,G) = \begin{cases}
        \prod_{\ell \in \leaf(T)}\hom(\lambda(\ell), G)                                                                           & \text{ if } |\elim(T)| = 0  \\
        \sum_{\overline{v} \in V(G)^{|\elim(T)|}}\prod_{\ell \in \leaf(T)}\hom(\lambda(\ell), G[\elim(T, \ell) \to \overline{v}]) & \text{ if } |\elim(T)| > 0.
    \end{cases}
    \end{equation}    
    If $T$ contains no elimination node, it is the product of the fully labeled graphs assigned to the leafs of $T$, so by \Cref{lem:hom:counts} we have $\hom(F,G) = \prod_{\ell \in \leaf(T)}\hom(\lambda(\ell), G)$.
    If $T$ contains the root $r$ as single elimination node deleting a label $z$, we have $F = F'[z \to \bot]$ with $z \in L_{F'}$ and $\hom(F',G) = \prod_{\ell \in \leaf(T)}\hom(\lambda(\ell), G)$ by induction.
    Again using \Cref{lem:hom:counts} this yields
    $$\hom(F,G) = \sum_{v \in V(G)} \hom(F', G[z \to v]) = \sum_{v \in V(G)} \prod_{\ell \in \leaf(T)}\hom(\lambda(\ell), G[z \to v]).$$
    Now let $|\elim(T)| > 1$.
    If $r$ is a product node with children $t_1,t_2 \in V(T)$, we have $F = \lambda(t_1)\lambda(t_2)$ and $t_2$ is a leaf.
    Then by induction
    \begin{align*}
      \hom(F,G) & = \hom(\lambda(t_1)\lambda(t_2), G)                                                                                                                              \\
                & = \hom(\lambda(t_2), G) \sum_{\overline{v} \in V(G)^{|\elim(T[t_1])|}}\prod_{\ell \in \leaf(T[t_1])}\hom(\lambda(\ell), G[\elim(T[t_1], \ell) \to \overline{v}]) \\
                & =  \sum_{\overline{v} \in V(G)^{|\elim(T)|}}\prod_{\ell \in \leaf(T)}\hom(\lambda(\ell), G[\elim(T, \ell) \to \overline{v}]).
    \end{align*}
    Here $T[t_1]$ is the induced subtree with root $t_1$. 
    If $r$ is an elimination node with single child $s$ deleting a label $z$ we have $F = F'[z \to \bot]$ with $z \in L_{F'}$ and again by induction we have
    \begin{align*}
      \hom(F,G) & = \sum_{w \in V(G)} \hom(F', G[z \to w])                                                                                                                        \\
                & = \sum_{w \in V(G)} \sum_{\overline{v} \in V(G)^{|\elim(T[s])|}}\prod_{\ell \in \leaf(T[s])}\hom(\lambda(\ell), G[\elim(T[s], \ell) \to \overline{v}][z \to w]) \\
                & = \sum_{\overline{v} \in V(G)^{|\elim(T)|}}\prod_{\ell \in \leaf(T)}\hom(\lambda(\ell), G[\elim(T, \ell) \to \overline{v}])
    \end{align*}
    completing the proof of \eqref{eq:path:hom-count}.

    Next, we show by induction on the structure of $T$ that there exists a non-counting formula $\varphi_T \in {\land}\mathsf{C}^{(k_1,k_2)}$ such that
    for all $\overline{v} \in V(G)^{|\elim(T)|}$ it holds that
    \begin{equation}
      \label{eq:caterpillar-formula}
      \tag{$\ast\ast$}
      G, \overline{v} \models \varphi_T \text{ if and only if } \prod_{\ell \in \leaf(T)} \hom(\lambda(\ell), G[\elim(T, \ell) \to \overline{v}]) = 1.
    \end{equation}
  
    For the base case, let $\ell$ be a leaf of $T$ and define
    $$\varphi_T^{\ell} \coloneqq \bigwedge_{\nu_{\lambda(\ell)}(z_i) = \nu_{\lambda(\ell)}(z_j)} z_i=z_j \land \bigwedge_{\nu_{\lambda(\ell)}(z_i) \nu_{\lambda(\ell)}(z_j) \in E({\lambda(\ell)})} E(z_i,z_j).$$
    
    For the inductive step, consider the following two cases:
    \begin{enumerate}
      \item If $t \in V(T)$ is a product node and $s_1,s_2 \in V(T)$ are the children of $t$ with $\lambda(t)=\lambda(s_1)\lambda(s_2)$,
    we define $\varphi_T^{s_1s_2} \coloneqq \varphi_T^{s_1} \land \varphi_T^{s_1}$.
    \item If $t \in V(T)$ is the $i$-th elimination node in $T$ with a single child $s$,
    then there exists a label $z \in [x_{k_1}, y_{k_2}]$ such that $\lambda(t)$ is obtained by deleting the label $z$ in $\lambda(s)$. We set
    $\varphi_T^{t} \coloneqq \exists z (z = w_i \land \varphi_T^{s})$.
    \end{enumerate}
  
    Finally, we define the formula~$\varphi_m^F \coloneqq \exists^m (w_1,\dotsc,w_{|\elim(T)|}) \varphi_T$ which has the desired property by combining \eqref{eq:path:hom-count} and \eqref{eq:caterpillar-formula}.
  \end{proof}

Next, we aim to prove that also every property definable in ${\land}\mathsf{C}^{(k_1,k_2)}$ can be modeled by counting homomorphisms from $\mathcal{LP}^{(k_1,k_2)}$.
In fact, the number of solutions to a non-counting formula in a graph can be expressed by counting homomorphisms from linear combinations:
\begin{lemma}\label{lem:path:model:combination}
    Let $k_1+k_2 \in \mathbb{N}_+$ and $\psi(w_{\ell_1},\dotsc,w_{\ell_m}, \overline{z}) \in {\land}\mathsf{C}^{(k_1,k_2)}$ be a non-counting formula.
    Then there exists a linear combination $\mathfrak{F}_{\psi} \in \mathbb{R}\mathcal{LP}^{(k_1,k_2)}$ such that that for all labeled graphs $G$ we have
    $$\hom(\mathfrak{F}_{\psi}, G) = \begin{cases}
            |\{\overline{v} \in V(G)^m : G, \overline{v} \models \psi\}| & \text { if } \free(\psi) \cap \mathcal{W} \neq \emptyset \\
            [G \models \psi]                                             & \text{ otherwise }
        \end{cases} $$
\end{lemma}
  \begin{proof}
    By \Cref{prop:primitive} we assume that the formula $\psi$ is primitive.
    We construct $\mathfrak{F}_{\psi}$ such that $\free(\psi) \cap [x_{k_1}, y_{k_2}] = L_{\mathfrak{F}_{\psi}}$ by induction on the structure of $\psi$:
  
    If $\psi = (z_i = z_j)$ we let $\mathfrak{F} \coloneqq K_1$ with the single vertex labeled by $z_i$ and $z_j$.
  
    If $\psi = E(z_i, z_j)$ and $i \neq j$ we let $\mathfrak{F} \coloneqq K_2$ with the endpoints being labeled with $z_i$ and $z_j$ respectively.
    For $i = j$ we let $\mathfrak{F}$ be the linear combination with no non-zero coefficients.
  
    If $\psi = \neg\psi'$ for an atomic formula $\psi'$ there exists $\mathfrak{F}_{\psi'} \in \mathbb{R}\mathcal{LP}^{(k_1,k_2)}$ with $\hom(\mathfrak{F}_{\psi'}, G) = [G \models \psi']$ by induction.
    We then let $\mathfrak{F}_{\psi} \coloneqq -1 \mathfrak{F}_{\psi'} + I_{L_{\mathfrak{F}_{\psi'}}}$.
  
    If $\psi = \psi_1 \land \psi_2$ there exist $\mathfrak{F}_{\psi_1}, \mathfrak{F}_{\psi_2} \in \mathbb{R}\mathcal{LP}^{(k_1,k_2)}$ with $\mathfrak{F}_{\psi_1} = \sum_{i=1}^t c^1_iF^1_i, \mathfrak{F}_{\psi_2} = \sum_{i=1}^s c^2_iF^2_i$ by induction.
    If $\psi_1$ and $\psi_2$ are both quantifier-free, by induction all graphs $F^1_i,F^2_j$ are fully labeled and hence $F^1_iF^2_j \in \mathcal{LP}^{(k_1,k_2)}$ for $i \in [t], j \in[s]$.
    Thus, in this case we set $\mathfrak{F}_\psi \coloneqq \mathfrak{F}_{\psi_1} \mathfrak{F}_{\psi_2} = \sum_{i=1}^t \sum_{j=1}^s c^1_i c^2_j F^1_iF^2_j$.
    Since $\psi$ is primitive and a restricted conjunction, at most one of the formulas $\psi_1$ and $\psi_2$ contains a quantifier.
    Again by induction, only one of the linear combinations $\mathfrak{F}_{\psi_1}, \mathfrak{F}_{\psi_2}$ does not contain only fully labeled graphs and hence $F^1_iF^2_j \in \mathcal{LP}^{(k_1,k_2)}$ for $i \in [t], j \in[s]$.
    Therefore we can again set $\mathfrak{F}_\psi \coloneqq \mathfrak{F}_{\psi_1} \mathfrak{F}_{\psi_2}$.
  
    If $\psi = \exists z (z = w_{\ell_m} \land \psi')$ there exists $\mathfrak{F}_{\psi'}$ by induction.
    We obtain $\mathfrak{F}_{\psi}$ from $\mathfrak{F}_{\psi'}$ by deleting all occurrences of the label $z$ in graphs in $\mathfrak{F}_{\psi'}$.
    Then we have
    \begin{align*}
        \hom(\mathfrak{F}_{\psi}, G)
        & = \sum_{i=1}^t c_i \hom(F_i[z \to \bot], G) \\
        & = \sum_{i=1}^t \sum_{u \in V(G)} c_i \hom(F_i, G[z \to u]) \\
        & = \sum_{u \in V(G)} \hom(\mathfrak{F}_{\psi'}, G[z \to u]).
    \end{align*}
    If $\psi'$ does not contain a free variable from $\mathcal{W}$ this sum evaluates to $|\{v \in V(G) : G ,v\models \psi'(z)\}|$.
    If $\psi'=\psi'(w_{\ell_1},\dotsc,w_{\ell_{m-1}})$ contains free variables $\overline{w} = [w_{\ell_1},\dotsc,w_{\ell_{m-1}}] \in \mathcal{W}$ we have $\hom(\mathfrak{F}_{\psi'}, G) = |\{\overline{v} \in V(G)^{m-1} : G,\overline{v} \models \psi'(\overline{w})\}|$
    and thus 
    \begin{align*}
      \hom(\mathfrak{F}_{\psi}, G) &= \sum_{u \in V(G)} |\{\overline{v} \in V(G)^{m-1} : G[z \to u], \overline{v} \models \psi'(\overline{w})\}| \\
      &= |\{\overline{v} \in V(G)^m : G, \overline{v} \models \psi(\overline{w}w_{\ell_m})\}|.
    \end{align*}
  \end{proof}

Combining the two previous results allows us to prove that homomorphism indistinguishability over the class $\mathcal{P}^{(k_1,k_2)}$ is the same as logical equivalence with restricted conjunction and requantification.

\begin{theorem} \label{thm:path:hom:ind}
    For $k_1+k_2 \in \mathbb{N}_+$ and graphs $G,H$ the following are equivalent:
    \begin{itemize}
        \item $G$ and $H$ are homomorphism indistinguishable over the class $\mathcal{P}^{(k_1,k_2)}$.
        \item $G$ and $H$ are ${\land}\mathsf{C}^{(k_1,k_2)}$-equivalent.
    \end{itemize}
\end{theorem}
\begin{proof}
    Assume there exists a graph $F \in \mathcal{P}^{(k_1,k_2)}$ such that $\hom(F, G) \neq \hom (F, H)$.
    By \Cref{lem:path:hom:count:definable} for every $m \in \mathbb{N}$ there exists $\varphi^F_m \in {\land}\mathsf{C}^{(k_1,k_2)}$ such that $G \models \varphi^F_m$ exactly if $\hom(F,G) = m$.
    Thus, $G \models \varphi^F_{\hom(F, G)}$ and $G \centernot\models \varphi^F_{\hom(F, G)}$.

    Now assume that there exists a sentence $\varphi \in {\land}\mathsf{C}^{(k_1,k_2)}$ such that $G \models \varphi$ and $H \centernot\models \varphi$.
    If $\varphi = \bigvee_{i \in [m]} \varphi_i$ then there already exists $i \in [m]$ such that $G \models \varphi_i$ and $H \centernot\models \varphi_i$.
    If $\varphi = \exists^{n} (w_{\ell_1},\dotsc,w_{\ell_m}) \psi(w_{\ell_1},\dotsc,w_{\ell_m})$ for a non-counting formula $\psi \in {\land}\mathsf{C}^{(k_1,k_2)}$,
    by \Cref{lem:path:model:combination} there exists $\mathfrak{F}_\psi$ with
    \[ \hom(\mathfrak{F}_{\psi}, G) = |\{\overline{v} \in V(G)^m : G, \overline{v} \models \psi\}| \neq |\{\overline{v} \in V(H)^m : H, \overline{v} \models \psi\}|=  \hom(\mathfrak{F}_{\psi}, H) ~ \qedhere \]
\end{proof}

\subsection{Construction trees}

Following the proof from \cite{Fluck2024}, the proof technique from the previous section can be used to adapt all constructions for the logic $\mathsf{C}^k_q$ with unrestricted conjunction and also preserve requantification.

\begin{lemma} \label{lem:tree:hom:count:definable}
    Let $F \in \mathcal{LT}^{(k_1,k_2)}_q$ and let $m \geq 0$.
    Then there exists a formula $\varphi^F_m \in \mathsf{C}^{(k_1,k_2)}_q$ with $\free(\varphi^F_m ) = L_F$ such that for each $[x_{k_1}, y_{k_2}]$-labeled graph $G$ with $L_F \subseteq L_G$ it holds that
    \begin{center}
      $G, \nu_G \models \varphi^F_m$ if and only if $\hom(F, G) = m$.
    \end{center}
  \end{lemma}
  \begin{proof}
    Let $(T, \lambda, r)$ be a $(k_1,k_2,q)$-construction tree of $F$.
    We proceed by induction on the structure of  $(T,\lambda,r)$.
    If $\ell \in V(T)$ is a leaf the graph $\lambda(\ell) = F$ is fully labeled.
    Thus, there exists at most one homomorphism from $F$ to $G$ and we set
    $$ \varphi^{\ell}_m \coloneqq \begin{cases}
        \bot                                                                                                    & \text{ if } m > 1 \\
        \bigwedge_{\nu_F(z_i) = \nu_F(z_j)} z_i=z_j \land \bigwedge_{\nu_F(z_i) \nu_F(z_j) \in E(F)} E(z_i,z_j) & \text{ if } m = 1 \\
        \neg \varphi^{\ell}_1                                                                                   & \text{ if } m=0
      \end{cases} $$
    Then $\qr(\varphi^{\ell}_m)=0$, only variables from $L_F \subseteq [x_{k_1}, y_{k_2}]$ are used in $\varphi^{\ell}_m$, and no variable is (re)quantified.
  
    Next, assume that $t \in V(T)$ is a product node.
    It suffices to treat the case where $t$ has two children $s_1,s_2 \in V(T)$ with $\lambda(t) = \lambda(s_1) \lambda(s_2)$.
    By the induction hypothesis, for all $m \in \mathbb{N}$ there exist formulas $\varphi^{s_1}_m, \varphi^{s_2}_m \in \mathsf{C}^{(k_1,k_2)}_q$ such that $G \models \varphi^{s_i}_m$ if and only if $\hom(\lambda(s_i), G) = m$.
    We set $\varphi^t_0 \coloneqq \varphi^{s_1}_0 \lor \varphi^{s_2}_0$ and for $m>0$ there exist $m_1,m_2 >0$ with $m=m_1m_2$, so we set $\varphi^t_{m_1m_2} = \varphi^{s_1}_{m_1} \land \varphi^{s_2}_{m_2}$.
    As no quantification was introduced we have $\varphi^t_m \in \mathsf{C}^{(k_1,k_2)}_q$ and more precisely $\qr(\varphi^t_m) = \qr(\varphi^{s_i}_m)$.
  
    Finally, let $t \in V(T)$ be an elimination node with a single child $s$.
    Then there exists a label $z \in [x_{k_1}, y_{k_2}]$ such that $\lambda(t)$ is obtained by deleting the label $z$ in $\lambda(s)$.
    Again, by the induction hypothesis for each $m \in \mathbb{N}$ there exists a formula $\varphi^{s}_m \in \mathsf{C}^{(k_1,k_2)}_q$ such that $G \models \varphi^{s}_m$ if and only if $\hom(\lambda(s), G) = m$.
    For $m>0$ we set $\varphi^{t}_m \coloneqq \bigvee_{m \in M} \bigl( \exists^{= c} z \neg\varphi^s_0 \land \bigwedge_{i \in [t]} \exists^{= c_i} z \varphi^s_{m_i} \bigr)$
    where $M$ is the set of all formal decompositions $\sum_{i=1}^t c_i m_i = m$ and $c = \sum_{i=1}^t c_i$. Also, we set $\varphi^t_0 \coloneqq \forall z \varphi^s_0$.
    Then it holds that
    $$\free(\varphi^t_m) = \free(\varphi^s_m) \setminus \{z\} = L_{\lambda(s)} \setminus \{z\} = L_{\lambda(t)}$$
    and if $z \in [y_{k_2}]$ it is not requantified as $z \in L_{\lambda(s)} = \free(\varphi^s_m)$ by induction.
  
    For the root $r$ of $T$ we let $\varphi^F_m \coloneqq \varphi^r_m$. The quantifier-rank is equal to the maximum number of elimination nodes on a branch in $T$, so $\varphi_m \in \mathsf{C}^{(k_1,k_2)}_q$.
  \end{proof}
  
  For a logical formula $\varphi$ and $n \in \mathbb{N}$ we say that $\mathfrak{F}$ \emph{models} $\varphi$ for graphs of order $n$ if $L_{\mathfrak{F}} = \free(\varphi)$ and for all labeled graphs $G$ of order $n$ it holds that
  $\hom(\mathfrak{F}, G) = 1$ if $G, \nu_G \models \varphi$ and $\hom(\mathfrak{F}, G) = 0$ otherwise.
  
  \begin{lemma}[{\cite[Lemma 5]{Dvorak2010}}]\label{lem:interpolation}
    Let $\mathcal{F}$ be a class of graphs that is closed under taking products and let $\mathfrak{F} \in \mathbb{R}\mathcal{F}$.
    Then for all disjoint finite sets $S^-,S^+ \subseteq \mathbb{R}$ there exists a linear combination $\mathfrak{F}[S^-, S^+] \in \mathbb{R}\mathcal{F}$ such that for every graph $G$ it holds that
    \begin{itemize}
      \item $\hom(\mathfrak{F}[S^-, S^+], G) = 1$ if $\hom(\mathfrak{F}, G) \in S^+$
      \item $\hom(\mathfrak{F}[S^-, S^+], G) = 0$ if $\hom(\mathfrak{F}, G) \in S^-$
    \end{itemize}
  \end{lemma}
  
  \begin{lemma} \label{lem:tree:model:combination}
    Let $k_1+k_2,q \geq 1$ and $\varphi \in \mathsf{C}^{(k_1,k_2)}_q$.
    Then for every $n \in \mathbb{N}$ there exists a linear combination $\mathfrak{F} \in \mathbb{R}\mathcal{LT}^{(k_1,k_2)}_q$ that models $\varphi$ for graphs of order $n$.
  \end{lemma}
  \begin{proof}
    We proceed by induction on the structure of the formula $\varphi \in \mathsf{C}^{(k_1,k_2)}_q$.
  
    If $\varphi = (z_i = z_j)$ we let $\mathfrak{F} \coloneqq K_1$ with the single vertex labeled by $i$ and $j$.
  
    If $\varphi = E(z_i, z_j)$ and $i \neq j$ we let $\mathfrak{F} \coloneqq K_2$ with the endpoints being labeled with $i$ and $j$ respectively.
    For $i = j$ we let $\mathfrak{F}$ be the linear combination with no non-zero coefficients.
  
    If $\varphi = \neg\psi$ there exists $\mathfrak{F}_{\psi} \in \mathbb{R}\mathcal{LT}^{(k_1,k_2)}_q$ that models $\psi$ for graphs of order $n$ by induction.
    We then let $\mathfrak{F} \coloneqq \mathfrak{F}_{\psi}[\{1\}, \{0\}]$ and since $\mathcal{LT}^{(k_1,k_2)}_q$ is closed under taking products from \Cref{lem:interpolation} it follows that $\mathfrak{F} \in \mathbb{R}\mathcal{LT}^{(k_1,k_2)}_q$ models $\varphi$ for graphs of order $n$.
  
    If $\varphi = \psi_1 \lor \psi_2$ there exist $\mathfrak{F}_{\psi_1}, \mathfrak{F}_{\psi_2} \in \mathbb{R}\mathcal{LT}^{(k_1,k_2)}_q$ modeling the respective formulas for graphs of order $n$ by induction.
    We set $\mathfrak{F} \coloneqq (\mathfrak{F}_{\psi_1} + \mathfrak{F}_{\psi_2})[\{0\}, \{1,2\}]$ and again by \Cref{lem:interpolation} the claim follows.
  
    If $\varphi = \exists^{\geq t} z \psi$ by induction there exists $\mathfrak{F}_{\psi} \in \mathbb{R}\mathcal{LT}^{(k_1,k_2)}_{q-1}$ that models $\psi$ for graphs of order $n$ since $\qr(\psi) = \qr(\varphi)-1$.
    Let $\mathfrak{F}'_{\psi}$ be the linear combination obtained from $\mathfrak{F}_{\psi}$ by deleting all the labels $z$ occurring in graphs in $\mathfrak{F}_{\psi}$.
    We then set $\mathfrak{F} \coloneqq \mathfrak{F}'_{\psi}[\{0,...,t-1\}, \{t,...,n\}]$, which models $\varphi$ for graphs of order $n$.
    If $z \in [y_{k_2}]$ it holds that $z \in \free(\psi) = L_{\mathfrak{F}_\psi}$ (or $z$ does not occur in $\psi$) and hence $z$ was not deleted in the construction tree of any graph in $\mathfrak{F}_{\psi}$.
    Thus, $\mathfrak{F} \in \mathbb{R}\mathcal{LT}^{(k_1,k_2)}_{q}$.
  \end{proof}

\begin{theorem}\label{thm:tree:hom:ind}
    For $k_1+k_2, q \in \mathbb{N}_+$ and graphs $G,H$ the following are equivalent:
    \begin{itemize}
        \item $G$ and $H$ are $\mathsf{C}^{(k_1,k_2)}_q$-equivalent.
        \item $G$ and $H$ are homomorphism indistinguishable over the class $\mathcal{T}^{(k_1,k_2)}_q$.
    \end{itemize}
\end{theorem} 
  \begin{proof}
    First, assume there exists a graph $F \in \mathcal{T}^{(k_1,k_2)}_q$ such that $\hom(F, G) \neq \hom (F, H)$.
    By \Cref{lem:tree:hom:count:definable} for every $m \in \mathbb{N}$ there exists $\varphi^F_m \in \mathsf{C}^{(k_1,k_2)}_q$ such that $G \models \varphi^F_m$ exactly if $\hom(F,G) = m$.
    Thus, $G \models \varphi^F_{\hom(F, G)}$ and $G \centernot\models \varphi^F_{\hom(F, G)}$.
  
    For the converse assume there exists a sentence $\varphi \in \mathsf{C}^{(k_1,k_2)}_q$ such that $G \models \varphi$ and $G \centernot\models \varphi$.
    If $|G| \neq |H|$ it follows that $\hom(K_1,G) \neq \hom(K_1, H)$, so assume $|G| = |H| = n$.
    By \Cref{lem:tree:model:combination} there exists $\mathfrak{F}_\varphi = \sum_i c_i F^{\varphi}_i \in \mathbb{R}\mathcal{LT}^{(k_1,k_2)}_q$ modeling $\varphi$ for graphs of order $n$.
    Thus, it holds that $\hom(\mathfrak{F}_{\varphi}, G) \neq \hom(\mathfrak{F}_{\varphi}, H)$ and hence there exists $F_i^{\varphi} \in \mathcal{LT}^{(k_1,k_2)}_q$ such that $\hom(F_i^{\varphi}, G) \neq \hom(F_i^{\varphi}, H)$.
  \end{proof}

\subsection{Homomorphism distinguishing closedness} \label{subsec:hdc}

We follow the approach from \cite{Neuen2024} to combine pursuit-evasion and model-comparison games to prove
homomorphism distinguishing closedness for the classes $\mathcal{T}^{(k_1,k_2)}_q$ and $\dotcup\mathcal{P}^{(k_1,k_2)}$.

First, we introduce a reusability-restricted variant of the all-in-one bijective pebble game from \cite{Montacute2024} to characterize ${\land}\mathsf{C}^{(k_1,k_2)}_{\infty \omega}$-equivalence.
\begin{definition}
    Let $\mathcal{A}, \mathcal{B}$ be $\sigma$-structures and $k_1+k_2 \in \mathbb{N}_+$. The \emph{all-in-one bijective $(k_1,k_2)$-pebble game} $\ABP^{(k_1,k_2)}(\mathcal{A}, \mathcal{B})$ is defined as follows:

    The game is played by the two players \emph{Spoiler} and \emph{Duplicator} on the structures $\mathcal{A}$ and $\mathcal{B}$.
    During the first and only round of the game, the following steps are performed:
    \begin{enumerate}
        \item Spoiler chooses a sequence of pebbles $\overline{z} = (z_1,\dotsc,z_n) \in [x_{k_1}, y_{k_2}]^n$ such that each $y_j \in [y_{k_2}]$ occurs at most once in $\overline{z}$.
        \item Duplicator chooses a bijection $h_{\overline{z}} \colon V(\mathcal{A})^n \to V(\mathcal{B})^n$.
        \item Spoiler chooses $\overline{v} \in V(\mathcal{A})^n$ and defines the sequence $\overline{s} \coloneqq [(z_i, v_i)]_{i \in [n]}$.
        \item Duplicator responds with the sequence $\overline{d} \coloneqq [(z_i, h_{\overline{z}}(\overline{v})[i])]_{i \in [n]}$.
    \end{enumerate}
    Duplicator wins if for all $i \in [n]$ the function $\eta_i$ defined by setting $\eta_i(\last_z(\overline{s}[1,i])) \coloneqq \last_z(\overline{d}[1,i])$
    for each $z \in [x_{k_1}, y_{k_2}]$ is a partial isomorphism between $\mathcal{A}$ and $\mathcal{B}$.
\end{definition}

We extend the proof of \cite[Theorem 5.9.]{Montacute2024} to the setting with restricted requantification and obtain the following: 
\begin{theorem} \label{thm:logic:abp}
    Let $\mathcal{A}, \mathcal{B}$ be finite $\sigma$-structures and $k_1+k_2 \in \mathbb{N}_+$. The following are equivalent:
    \begin{enumerate}
        \item $\mathcal{A}$ and $\mathcal{B}$ are ${\land}\mathsf{C}^{(k_1,k_2)}_{\infty \omega}$-equivalent. \label{inifinitary-equiv}
        \item $\mathcal{A}$ and $\mathcal{B}$ are ${\land}\mathsf{C}^{(k_1,k_2)}$-equivalent. \label{finite-equiv}
        \item Duplicator has a winning strategy for $\ABP^{(k_1,k_2)}(\mathcal{A},\mathcal{B})$. \label{duplicator-winning}
    \end{enumerate}
\end{theorem} 
\begin{proof}
  We restrict ourselves to the case of undirected graphs here for notational convenience.
  The proof is the same for $\sigma$-structures by replacing the edge relation with every $R \in \sigma$. 

    The implication \ref{inifinitary-equiv} $\Longrightarrow$ \ref{finite-equiv} is immediate. 

    \ref{finite-equiv} $\Longrightarrow$ \ref{duplicator-winning}: 
    We argue by contraposition and assume that Spoiler has a winning strategy for the game $\ABP^{(k_1,k_2)}(G,H)$. 
    Let $\overline{p}=(p_1,\dotsc,p_n)$ be the pebble sequence Spoiler chooses in the beginning of the game. 
    For a graph $F$ and $\overline{v} \in V(F)^n$ we set $\overline{pv} \coloneqq [(p_i,v_i)]_{i \in [n]}$
    We construct a non-counting formula $\mathsf{tp}_{\overline{pv}}(w_1,\dotsc,w_n) \in {\land}\mathsf{C}^{(k_1,k_2)}$
    such that for every graph $F'$ and $\overline{u} \in V(F')^n$ it holds $F' \models \mathsf{tp}_{\overline{pv}}(\overline{u})$
    if and only if for all $i \in [n]$ the function $\eta_i$ defined by setting $\eta_i(\last_z(\overline{pv}[i])) \coloneqq \last_z(\overline{pu}[i])$
    for each $z \in [x_{k_1}, y_{k_2}]$ is a partial isomorphism.   
    For $i \in [n]$ define the \emph{literal diagram} of $\overline{pv}[1,i]$ by induction as
    \begin{align*}
      \mathsf{diag}_{\overline{pv}[1,i]} \coloneqq & \{E(z_\ell, z_j) : \last_{z_\ell}(\overline{pv}[1,i]) = v_{i_\ell}, \last_{z_j}(\overline{pv}[1,i+1]) = v_{i_j},  v_{i_\ell}v_{i_j} \in E(F)\} \\
      & \cup \{\neg E(z_\ell, z_j) : \last_{z_\ell}(\overline{pv}[1,i]) = v_{i_\ell}, \last_{z_j}(\overline{pv}[1,i]) = v_{i_j},  v_{i_\ell}v_{i_j} \not\in E(F)\}
    \end{align*}
    Next, let $m \in [n-1]\cup\{0\}$ be the largest index such that after no pebble occurs more than once in $\overline{p}[m+1,n]$. 
    By induction on $j \in [m-1]$ we define the formula $\varphi_j(w_1,\dotsc,w_j,z_1\dotsc,z_{j-1})$ by $\varphi_0 \coloneqq \top$ and 
    $$\varphi_{j+1} \coloneqq \exists z_j \Bigl(z_j = w_j \land \bigwedge_{\psi \in \mathsf{diag}_{\overline{pv}[j]}} \psi \land \varphi_j\Bigr)$$
    Finally, we define 
    $$\mathsf{tp}_{\overline{pv}} \coloneqq \exists z_{m+1} \ldots z_{n} \Bigl(\bigwedge_{j \in \{m+1,\dots,n\}}z_j=w_j \land \bigwedge_{\psi \in \mathsf{diag}_{\overline{pv}}} \psi \land \varphi_m\Bigr)$$
    Note that in $\mathsf{tp}_{\overline{pv}[i]}$ no variable $y \in [y_{k_2}]$ is requantified because 
    no corresponding pebble is reused in $\overline{p}$. Also, each set $\mathsf{diag}_{\overline{pv}[j]}$ is finite, so $\mathsf{tp}_{\overline{pv}[i]} \in {\land}\mathsf{C}^{(k_1,k_2)}$.
    In $\ABP^{(k_1,k_2)}(G,H)$, Duplicator chooses a bijection $h \colon V(G)^n \to V(H)^n$ and Spoiler responds with a choice $\overline{v}\in V(G)^n, h(\overline{v}) \in V(H)^n$.
    Since Spoiler wins the game, it holds that for some $i \in [n]$, the function defined by setting 
    $\eta_i(\last_z(\overline{pv}[1,i])) \coloneqq \last_z(\overline{ph(\overline{v})}[i])$ for each $z \in [x_{k_1}, y_{k_2}]$ is not a partial isomorphism.
    This in turn yields $G \models \mathsf{tp}_{\overline{pv}[1,i]}$ but $H \centernot\models \mathsf{tp}_{\overline{ph(\overline{v})}[1,i]}$ (or vice versa).  

    \ref{duplicator-winning} $\Longrightarrow$ \ref{inifinitary-equiv}: 
    We show that for every sentence $\varphi \in {\land}\mathsf{C}_{\infty \omega}^{(k_1,k_2)}$ with $G \models \varphi$ it also holds $H \models \varphi$ and vice versa (which follows by symmetry). 
    By \Cref{prop:primitive} we may assume that $\varphi$ is a disjunction of primitive formulas of the form $\exists^n \overline{w} \psi(\overline{w})$.
    We let $\overline{w} = (w_{\ell_1},\dotsc,w_{\ell_m})$ and assume that in $\ABP^{(k_1,k_2)}(G,H)$ Spoiler picks the pebble sequence $\overline{z} = (z_1,\dotsc,z_m)$
    such that $z_i = w_{\ell_i}$ is the equality guarding the quantification $\exists z_i$ in $\varphi$. 
    Then each $y_j \in [y_{k_2}]$ occurs at most once in this pebble sequence since it is not requantified in $\varphi$. 
    Let $h_{\overline{z}} \colon V(G)^m \to V(H)^m$ be the response from Duplicator's winning strategy.
    This function is a partial isomorphism on the actively pebbled vertices according to the sequence $\overline{z}$.
    Thus, For every $\overline{v} \in V(G)^m$ it holds that $G \models \psi(\overline{v})$ exactly if $H \models \psi(h_{\overline{z}}(\overline{v}))$ by the isomorphism-invariance of first-order logic.
    By assumption, there exist exactly $n$ such tuples $\overline{v} \in V(G)^m$ and since $h_{\overline{z}}$ is a bijection, also exactly $n$ such tuples $h_{\overline{z}}(\overline{v}) \in V(H)^m$.
    Therefore, we also have that $H \models \varphi$. 
\end{proof}
The proof shows that the theorem also holds for infinite structures if we omit \Cref{finite-equiv}.

The next lemma is the key technical ingredient to establish our first homomorphism distinguishing closedness result.
But more generally, it shows that the capability of $G$ to be decomposed in a path-like fashion with reusability constraints provides a lower bound for the distinguishability of the CFI graphs $X(G), \widetilde{X}(G)$ by ${\land}\mathsf{C}^{(k_1,k_2)}_{\infty \omega}$.
The proof idea is that the position of the fugitive in $\NS^{(k_1,k_2)}(G)$ corresponds to the position where the difference of the CFI graphs is moved to in $\ABP^{(k_1,k_2)}(X(G), \widetilde{X}(G))$ by \Cref{lem:cfi:twist:iso}.
Thus, the invisibility of the fugitive corresponds to the fact that Spoiler has to fix their entire strategy at once.

\begin{lemma} \label{lem:cfi:path:game}
    Let $k_1+k_2 \in \mathbb{N}_+$ and $G$ be a connected graph. If the fugitive has a winning strategy for the game $\NS^{(k_1,k_2)}(G)$, then Duplicator has a winning strategy for $\ABP^{(k_1,k_2)}(X(G), \widetilde{X}(G))$.
\end{lemma} 
\begin{proof}
  Given a winning strategy for the fugitive in $\NS^{(k_1,k_2)}(G)$, we provide a winning strategy for Duplicator in $\ABP^{(k_1,k_2)}(X(G), \widetilde{X}(G))$.
  By \Cref{lem:cfi:isomorphic} we may assume $\widetilde{X}(G) = X_{\{u\}}(G)$ for some $u \in V(G)$ and that in $\NS^{(k_1,k_2)}(G)$ the fugitive is initially placed on $u$.

  At the beginning of the game, Spoiler chooses a sequence $\overline{p} = (p_1,\dotsc,p_n) \in [x_{k_1}, y_{k_2}]^n$ of pebbles.
  We construct a bijection $h_{\overline{p}} \colon V(X(G))^n \to V(X_{\{u\}}(G))^n$ for the winning strategy of Duplicator.
  Given the sequence of pebbles $\overline{p}$ and a sequence of vertices $((v_1, S_1),\dotsc,(v_n, S_n)) \in V(X(G))^n$, we let $\overline{s} = [p_i, (v_i, S_i)]_{i \in [n]}$.
  Let $f(\overline{s}[1,i])$ be the vertex which the fugitive escapes to when the searchers are placed according to the base vertices of $\overline{s}[1,i]$,
  i.e., when $\gamma(z) = \rho(\last_z(\overline{s}[1,i]))$ for each $z \in [x_{k_1},y_{k_2}]$.
  Since the fugitive has a winning strategy by assumption, for each $i \in [n-1]$ there exists a path $P$ from $f(\overline{s}[1,i])$ to $f(\overline{s}[1,i+1])$ in $G$ such that not vertex of $P$ is occupied by a searcher.
  By \Cref{lem:cfi:twist:iso} there exists an isomorphism $\varphi_{f(\overline{s}[1,i+1]), f(\overline{s}[1,i])} : X_{f(\overline{s}[1,i+1])} \to X_{f(\overline{s}[1,i])}$ which preserves gadgets and is the identity on every gadget outside of $P$.
  For each sequence $\overline{s} = [p_i, (v_i, S_i)]_{i \in [n]}$ we set
  $$\psi_1 \coloneqq \varphi_{f(\overline{s}[1]), u} \text{ and } \psi_{i+1} \coloneqq \psi_i \circ \varphi_{f(\overline{s}[1,i+1]), f(\overline{s}[1,i])}.$$
  We then define
  $$h_{\overline{p}} \colon ((v_1, S_1),\dotsc,(v_n, S_n)) \mapsto (\psi_1(v_1, S_1), \dotsc, \psi_n(v_n, S_n))$$
  and first show that $h_{\overline{p}}$ is a bijection by induction on $n$.
  For $n=1$ consider the function $h_{\overline{p}} \colon (v,S) \mapsto \varphi_{f(v), u}(v, S)$.
  It holds $f(\overline{s}[1]) \neq v_1$ and $\rho(h_{\overline{p}}(v_1,S)) = \rho(\varphi_{f(\overline{s}[1]), u}(v_1, S)) = v_1$. 
  If $v \in V(G)$ is fixed, the function $h_{\overline{p}}(v, \cdot) = \varphi_{f(v), u}(v, \cdot)$ is a bijection between $F_{\emptyset}(v)$ and $F_{u}(v)$.
  For $n > 1$ the restriction of $h_{\overline{p}}$ to $V(X(G))^{n-1}$ is a bijection to $V(X_{\{u\}}(G))^{n-1}$ by the inductive hypothesis.
  Now let $((v_1, S_1),\dotsc,(v_{n-1}, S_{n-1})) \in V(X(G))^{n-1}$ be fixed, then for each $v \in V(G)$ the map $(v,S) \mapsto \psi_n(v, S)$ is a bijection between $F_{\emptyset}(v)$ and $F_{u}(v)$.
  Thus, also $h_{\overline{p}}$ is a bijection and hence a valid choice for Duplicator.
  Now Spoiler chooses $\overline{a} \in V(X(G))^n$ and Duplicator responds with $h_{\overline{p}}(\overline{a})$.
  We show by induction on $i \leq n$ that the partial function $\eta_i \colon V(X(G)) \rightharpoonup V(\widetilde{X}_u(G))$ induced by mapping $ \last_z(\overline{s}[1,i]) \mapsto \last_z(\overline{d}[1,i])$ for each $z \in [x_{k_1}, y_{k_2}]$ is a partial isomorphism.
  For $i=1$ the claim holds trivially.
  Assume by induction that for $1 < i \leq n$ the function $\eta_{i-1}$ is a partial isomorphism.
  We distinguish two cases:

  First, assume $p_i$ occurs for the first time at index $i$ in $\overline{p}$, then $\eta_{i-1}$ and $\eta_i$ differ in at most one vertex $(v, S) \in V(X(G))$ with $(v,S) = \last_{p_i}(\overline{s}[i]) = a_i$.
  If $(v,S)$ occurs  at some position $j < i$ in $\overline{s}[i]$, then $\varphi_{f(\overline{s}[j-1]), f(\overline{s}[j])}(v, S) = \varphi_{f(\overline{s}[i-1]), f(\overline{s}[i])}(v, S)$, so $\eta_i$ is a partial isomorphism.
  Otherwise, it holds $\eta_{i-1}(v,S) = \bot$ and $\eta_i(v, S) = (v, S')$ for some $(v, S') \in F_u(v)$.
  The map $\varphi_{f(\overline{s}[i-1]), f(\overline{s}[i])}$ is an isomorphism between $\widetilde{X}_{f(\overline{s}[i-1])}(G) - F_{f(\overline{s}[i-1])}$ and $\widetilde{X}_{f(\overline{s}[i])}(G) - F_{f(\overline{s}[i])}$ and $v \neq f(\overline{s}[i])$, so also $\eta_i$ is a partial isomorphism.

  Second, assume there exists a position $j < i$ such that $p_j = p_i$ and let $j$ be the largest one. Then the pebble must be reusable, i.e. $p_i \in [x_{k_1}]$.
  Now $\eta_{i-1}$ and $\eta_i$ differ in two vertices $(v, S), (w,T) \in V(X(G))$ with $(v,S) = \last_{p_i}(\overline{s}[i]) = a_i$, $(w,T) = \last_{p_i}(\overline{s}[i-1]) = a_j$
  and $\eta_i(w,T) = \bot, \eta_i(v,S) = \eta_i(v,S')$. Then $v$ does not appear among $f(\overline{s}[j]), f(\overline{s}[j+1]),\dotsc,f(\overline{s}[i])$ and again each $\varphi_{f(\overline{s}[j']), f(\overline{s}[j'+1])}$
  is an isomorphism between $\widetilde{X}_{f(\overline{s}[j'])}(G) - F_{f(\overline{s}[j'])}$ and $\widetilde{X}_{f(\overline{s}[j'+1])}(G) - F_{f(\overline{s}[j'+1])}$ for $j\leq j'<i$.
  Then again, $\eta_i$ is a partial isomorphism by induction.
\end{proof}
 
We conjecture that using similar techniques also the reverse implication of this lemma can be shown, which would in particular yield 
that the pathwidth of a graph $G$ is exactly the minimum $k$ such that ${\land}\mathsf{C}^{k+1}$ distinguishes $X(G)$ and $\widetilde{X}(G)$.
However, to prove our next theorem \Cref{lem:cfi:path:game} suffices.

\begin{theorem} \label{thm:pathwidth:hdc}
    The classes $\mathcal{P}^{(k_1,0)}$ and $\dot{\cup}\mathcal{P}^{(k_1,k_2)}$ are homomorphism distinguishing closed.
\end{theorem}
\begin{proof}
    For $\mathcal{F} \in \{\mathcal{P}^{(k_1,0)}, \dot{\cup}\mathcal{P}^{(k_1,k_2)}\}$ the class $\mathcal{F}$ is closed under taking disjoint unions and summands.
    For every connected graph $G \notin \mathcal{F}$ the fugitive has a winning strategy for $\NS^{(k_1,k_2)}(G)$ by \Cref{thm:path:characterization}.
    By \Cref{thm:path:hom:ind} and \Cref{lem:cfi:path:game} this in turn yields $X(G) \equiv_{\mathcal{F}} \widetilde{X}(G)$.
    Finally, by \Cref{lem:sufficient:hdc} it follows that $\mathcal{F}$ is h.d. closed.
\end{proof}

Note that the class $\mathcal{P}^{(k_1,k_2)}$ is not homomorphism distinguishing closed since the relations $\equiv_{\dot{\cup}\mathcal{P}^{(k_1,k_2)}}$ and $\equiv_{\mathcal{P}^{(k_1,k_2)}}$
are identical, but $\mathcal{P}^{(k_1,k_2)} \subsetneq \dot{\cup}\mathcal{P}^{(k_1,k_2)}$.

After establishing the monotonicity of $\CR^{(k_1, k_2)}_q$ in \Cref{prop:non-reusable:monotone}, we can use the correspondence to the bijective pebble game \cite[Lemma 8]{Rassmann2025}
to recast the proof for $\mathcal{T}^{(k_1,k_2)}_q$. 

\begin{lemma}[\cite{Rassmann2025}] \label{lem:cfi:game}
    Let $k_1+k_2 \geq 1, q \geq 0$ and $G$ be a connected graph.
    If the robber has a winning strategy for $\CR^{(k_1,k_2)}_q(G)$ then $X(G) \equiv_{\mathcal{T}^{(k_1,k_2)}_q} \widetilde{X}(G)$.
\end{lemma}

\begin{theorem}\label{thm:treewidth:hdc}
    The class $\mathcal{T}^{(k_1,k_2)}_q$ is homomorphism distinguishing closed.
\end{theorem}
  \begin{proof}
    $\mathcal{T}^{(k_1,k_2)}_q$ is closed under taking disjoint unions and summands.
    For every $G \notin \mathcal{T}^{(k_1,k_2)}_q$ the robber has a winning strategy for $\CR^{(k_1,k_2)}_q(G)$ by \Cref{thm:tree:characterization} and \Cref{prop:non-reusable:monotone}.
    By \Cref{lem:cfi:game} this yields $X(G) \equiv_{\mathcal{T}^{(k_1,k_2)}_q} \widetilde{X}(G)$.
    Finally, by \Cref{lem:sufficient:hdc} it follows that $\mathcal{T}^{(k_1,k_2)}_q$ is h.d. closed.
  \end{proof}

\begin{remark}
    In the language of \cite{Roberson2022} the proofs of \Cref{thm:pathwidth:hdc} and \Cref{thm:treewidth:hdc} show that the respective classes are \emph{closed under weak oddomorphisms} by \cite[Theorem 3.13]{Roberson2022}.
\end{remark}

\subsection{Invariance of subgraph counts} \label{subsec:invariance}
As an application of the previous results of homomorphism distinguishing closedness and its relation to logic,
we provide characterizations of the logical invariance of subgraph counts with respect to hereditary graph structure as in \cite{Neuen2024}.

For graphs $G,F$ we denote by $\sub(F, G)$ the number of subgraphs of $G$ which are isomorphic to $F$. 
For two graphs $F$ and $H$ we say that $H$ is a homomorphic image of $F$ if there exists a surjective homomorphism $\varphi \colon V(F) \to V(H)$ such that $E(H) = \{\varphi(u)\varphi(v) : uv \in E(F)\}$.
We write $\operatorname{spasm}(F)$ for the set of homomorphic images of a graph $F$ containing exactly one representative from each isomorphism class.
For a logic $\mathsf{L}$ and a graph $F$ we say that the function $\sub(F, \cdot)$ is \emph{$\mathsf{L}$-invariant} if for all graphs $G$ and $H$ the implication 
$ G \equiv_{\mathsf{L}} H \Longrightarrow \sub(F,G) = \sub(F, H)$ holds.

To characterize for which graphs $\sub(F,\cdot)$ is invariant for a logic, we use the following lemmas relating homomorphism- and subgraph counts to the homomorphism distinguishing closure.

\begin{lemma}[{\cite[Section 5.2.3]{Lovasz2012},\cite{Curticapean2017}}] \label{lem:sub:hom}
    For every pair of graphs $G, F$ there exists a unique function $\alpha \colon \spasm(F) \to \mathbb{R}\setminus\{0\}$ such that
    $\sub(F,G) = \sum_{F_i \in \spasm(F)} \alpha(F_i) \hom(F_i, G)$.
\end{lemma}

\begin{lemma}[{\cite[Lemma 6]{Seppelt2024}}] \label{lem:lincomb}
    Let $\mathcal{F}$ be a class of graphs, $\mathcal{L}$ be a finite set of pairwise non-isomorphic graphs, and $\alpha \colon \mathcal{L} \to \mathbb{R}\setminus\{0\}$.
    If for all graphs $G$ and $H$ it holds that
    $$ G \equiv_{\mathcal{F}} H \Longrightarrow \sum_{L \in \mathcal{L}} \alpha(L) \cdot \hom(L, G)  = \sum_{L \in \mathcal{L}} \alpha(L) \cdot \hom(L, H)$$
    then $\mathcal{L} \subseteq \cl(\mathcal{F})$.
\end{lemma}

Finally we obtain the following theorem as a consequence of \Cref{thm:pathwidth:hdc,thm:treewidth:hdc}. 

\begin{theorem} \label{thm:invariance}
    Let $k_1+k_2, q \geq 1$ and $F$ be a graph.
    Then the following assertions hold:
    \begin{itemize}
        \item The function $\sub(F, \cdot)$ is ${\land}\mathsf{C}^{(k_1,k_2)}$-invariant if and only if $\spasm(F) \subseteq \dotcup\mathcal{P}^{(k_1,k_2)}$.
        \item The function $\sub(F, \cdot)$ is $\mathsf{C}^{(k_1,k_2)}_q$-invariant if and only if $\spasm(F) \subseteq \mathcal{T}^{(k_1,k_2)}_q$.
    \end{itemize}
\end{theorem}
\begin{proof}
    The proof is verbatim for both assertions, so let $\mathcal{F} \in \{\dotcup\mathcal{P}^{(k_1,k_2)}, \mathcal{T}^{(k_1,k_2)}_q\}$ and $G,H$ be graphs.
    First we assume that $\spasm(F) = \{F_1,\dotsc,F_{\ell}\} \subseteq \mathcal{F}$ and $G \equiv_{\mathcal{F}} H$ by \Cref{thm:path:characterization} and \Cref{thm:tree:characterization} respectively.
    By \Cref{lem:sub:hom} there exist coefficients $\alpha_1,\dotsc,\alpha_{\ell} \in \mathbb{R}$ such that for all graphs $G'$ it holds $\sub(F, G') = \sum_{i=1}^{\ell} \alpha_i \hom(F_i, G')$.
    Thus, we have
    $$ \sub(F,G) =  \sum_{i=1}^{\ell} \alpha_i \hom(F_i, G) = \sum_{i=1}^{\ell} \alpha_i \hom(F_i, H) = \sub(F, H)$$
    Now assume that $\sub(F, \cdot)$ is invariant for the respective logic.
    Again by \Cref{lem:sub:hom} there exists a unique function $\alpha \colon \spasm(F) \to \mathbb{R} \setminus \{0\}$ such that for all graphs $G$ it holds that $\sub(F,G) = \sum_{L \in \spasm(F)} \alpha(L) \cdot \hom(L, G)$.
    For graphs $G,H$ with $G \equiv_{\mathcal{F}} H$ it holds $\sub(F, G) = \sub(F, H)$ by assumption and hence
    $$\sum_{L \in \spasm(F)} \alpha(L) \cdot \hom(L, G) = \sum_{L \in \spasm(F)} \alpha(L) \cdot \hom(L, H).$$
    Since $\mathcal{F}$ is homomorphism distinguishing closed by \Cref{thm:pathwidth:hdc}, \Cref{thm:treewidth:hdc} and $\spasm(F)$ is a finite set of pairwise non-isomorphic graphs, \Cref{lem:lincomb} yields $\spasm(F) \subseteq \cl(\mathcal{F}) = \mathcal{F}$.
\end{proof}
\begin{remark} \label{rem:invariance}
    By \cite[Theorem 6]{Rassmann2025} the second assertion of this theorem characterizes which subgraph counts are detected after $q$ iterations of the $(k_1,k_2)$-dimensional oblivious Weisfeiler-Leman algorithm.
\end{remark}

  \section{A comonadic account of requantification} \label{sec:comonadic}

In this section, we present variations of the pebble-relation comonad $\mathbb{PR}_k$ introduced in \cite{Montacute2024} and the pebbling comonad $\mathbb{P}_k$ from \cite{Abramsky2017}, 
with the aim of capturing requantification as a logical resource from a categorical perspective. 
We emphasize that, in order to achieve this, it suffices to adapt the definitions of the associated universes $\mathbb{PR}_{(k_1,k_2)}\mathcal{A}$ and $\mathbb{P}_{(k_1,k_2)}\mathcal{A}$. 
This approach highlights the versatility of the pebble-relation and pebbling comonads, enabling concise proofs building on previous work despite the significant differences among the corresponding graph classes.

For a sequence $\overline{s} = [(z_1,a_1),\dotsc,(z_n,a_n)] \in ([x_{k_1}, y_{k_2}] \times V(\mathcal{A}))^n$ and $i \in [n]$ define $\pi_{\mathcal{A}}(\overline{s}, i) = z_i$.
When no index $i$ is given, we set $\pi_{\mathcal{A}}(\overline{s}) \coloneqq z_n$. 
\begin{definition}
    For a $\sigma$-structure $\mathcal{A}$ we define the $\sigma$-structure $\mathbb{PR}_{(k_1,k_2)}\mathcal{A}$ as follows:
    \begin{itemize}
        \item The universe of $\mathbb{PR}_{(k_1,k_2)}\mathcal{A}$ consists of all pairs $(\overline{s},i) = ([(z_1,a_1), \dotsc, (z_n,a_n)], i)$ for $n \in \mathbb{N}$, $i \in [n]$,
              and $(z_j,a_j) \in [x_{k_1}, y_{k_2}] \times V(\mathcal{A})$ for $j \in [n]$ such that every $y_j \in [y_{k_2}]$ appears at most once as pebble index in $\overline{s}$.
        \item The counit morphism $\varepsilon_{\mathcal{A}} \colon \mathbb{PR}_{(k_1,k_2)}\mathcal{A} \to \mathcal{A}$ is defined by $\varepsilon_{\mathcal{A}}([(z_1,a_1), \dotsc, (z_n,a_n)], i) \coloneqq a_i$
        \item For $R \in \sigma$ it holds $R^{\mathbb{PR}_{(k_1,k_2)}\mathcal{A}}((\overline{s}_1,i_1),\dotsc,(\overline{s}_m, i_m))$ exactly if there exists $\overline{s}$ such that
              \begin{itemize}
                  \item for all $j \in [m]$ it holds $\overline{s}_j = \overline{s}$, \hfill (equality)
                  \item $\pi_{\mathcal{A}}(\overline{s}, i_j)$ does not appear in $\overline{s}[{i_j} + 1,i]$ for $i = \max\{i_1,\dotsc,i_m\}$, and \hfill (active pebble)
                  \item $R^{\mathcal{A}}(\varepsilon_{\mathcal{A}}(\overline{s}_1, i_1),\dotsc,\varepsilon_{\mathcal{A}}(\overline{s}_m, i_m))$. \hfill (compatibility)
              \end{itemize}
    \end{itemize}
    For a $\sigma$-structure $\mathcal{B}$ and a homomorphism $f \colon \mathbb{PR}_{(k_1,k_2)}\mathcal{A} \to \mathcal{B}$ we define the \emph{coextension} 
    $f^* \colon \mathbb{PR}_{(k_1,k_2)}\mathcal{A} \to \mathbb{PR}_{(k_1,k_2)}\mathcal{B}$ of $f$ by setting $b_i = f([(z_1,a_1),\dotsc, (z_m,a_m)], i)$ for $i \in [m]$ and 
    $f^*([(z_1,a_1),\dotsc, (z_m,a_m)], i) \coloneqq ([(z_1,b_1),\dotsc, (z_m,b_m)], i).$
\end{definition} 

\begin{definition}
    For a $\sigma$-structure $\mathcal{A}$ we define the $\sigma$-structure $\mathbb{P}_{(k_1,k_2)}\mathcal{A}$ as follows:
    \begin{itemize}
        \item The universe of $\mathbb{P}_{(k_1,k_2)}\mathcal{A}$ consists of all sequences $\overline{s} \in ([x_{k_1},y_{k_2}] \times V(\mathcal{A}))^+$ such that
              every pebble $y_j \in [y_{k_2}]$ appears at most once as pebble index in $\overline{s}$.
        \item The counit morphism $\varepsilon_{\mathcal{A}} \colon \mathbb{P}_{(k_1,k_2)}\mathcal{A} \to \mathcal{A}$ is defined by
              $\varepsilon_{\mathcal{A}}([(z_1,a_1),\dotsc, (z_m,a_m)]) \coloneqq a_m$
        \item For $R \in \sigma$ it holds $R^{\mathbb{P}_{(k_1,k_2)}\mathcal{A}}(\overline{s}_1,\dotsc,\overline{s}_m)$ exactly if
              \begin{itemize}
                  \item $R^{\mathcal{A}}(\varepsilon_{\mathcal{A}}(\overline{s}_1),\dotsc,\varepsilon_{\mathcal{A}}(\overline{s}_m))$, \hfill (compatibility)
                  \item for $i,j \in [m]$ we have $\overline{s}_i \sqsubseteq \overline{s}_j$ or $\overline{s}_j \sqsubseteq \overline{s}_i$, and \hfill (comparability)
                  \item for $i,j \in [m]$ if $\overline{s}_i \sqsubseteq \overline{s}_j$ then $\pi_{\mathcal{A}}(\overline{s}_i)$ \\
                   does not occur as a first coordinate in $\overline{s}_+$ for $\overline{s}_i = \overline{s}_j\overline{s}_+$. \hfill (active pebble)
              \end{itemize}
    \end{itemize}
    For a $\sigma$-structure $\mathcal{B}$ and a homomorphism $f \colon \mathbb{P}_{(k_1,k_2)}\mathcal{A} \to \mathcal{B}$ the \emph{coextension} $f^*$ is defined by setting 
    $f^*([(z_1,a_1),\dotsc, (z_m,a_m)]) \coloneqq [(z_1,b_1),\dotsc, (z_m,b_m)]$ where $b_i = f([(z_1,b_1),\dotsc, (z_i,a_i)])$ for $i \in [m]$.
    For $q \in \mathbb{N}_+$ the structure $\mathbb{P}_{(k_1,k_2)}^q\mathcal{A}$ is defined as the substructure of $\mathbb{P}_{(k_1,k_2)}\mathcal{A}$ over the universe $([x_{k_1},y_{k_2}] \times V(\mathcal{A}))^{\leq q}$.
\end{definition}

To show that $\mathbb{PR}_{(k_1,k_2)}$ and $\mathbb{P}^q_{(k_1,k_2)}$ are again comonads it suffices to observe that the definitions of $R^{\mathbb{PR}_{(k_1,k_2)}\mathcal{A}}$,
$R^{\mathbb{P}^q_{(k_1,k_2)}\mathcal{A}}$, and the coextensions are invariant under the new reusability constraints. 
Thus, from \cite[Proposition 3.1]{Montacute2024} and \cite[Theorem 4]{Abramsky2017} we obtain the following:  
\begin{proposition} \label{prop:comonads}
    The triples $(\mathbb{P}^q_{(k_1,k_2)}, \varepsilon, (\cdot)^*)$ and $(\mathbb{PR}_{(k_1,k_2)}, \varepsilon, (\cdot)^*)$ are comonads in coKleisli form on the category~$\Str(\sigma)$.
\end{proposition} 

Utilizing the characterization of the classes $\mathcal{T}^{(k_1,k_2)}_q$ and $\mathcal{P}^{(k_1,k_2)}$ in terms of pebble forest covers, we now provide a categorical
account of reusability in path- and bounded depth tree decompositions. Similar results are called \emph{coalgebra characterization theorems} in the literature of game comonads \cite{Abramsky2021}.
We have defined our decompositions in terms of graphs rather than relational structures. The definition of forest covers can be adapted for structures, which yields the same 
as considering forest covers of the \emph{Gaifman graph} $G(\mathcal{A})$ of $\mathcal{A}$. 

\begin{theorem} \label{thm:coalgebras}
    For every finite $\sigma$-structure $\mathcal{A}$, there is a bijective correspondence between
    \begin{enumerate}
        \item \label{coalgebras:paths} $(k_1,k_2)$-pebble linear component forest covers of $G(\mathcal{A})$ and coalgebras $\alpha \colon \mathcal{A} \to \mathbb{PR}_{(k_1,k_2)}\mathcal{A}$.
        \item \label{coalgebras:trees} $(k_1,k_2)$-pebble forest covers of depth $q$ of $G(\mathcal{A})$ and coalgebras $\alpha \colon \mathcal{A} \to \mathbb{P}^q_{(k_1,k_2)}\mathcal{A}$.
    \end{enumerate}
\end{theorem}
\begin{proof}
  \Cref{coalgebras:paths}: We review the constructions from \cite[Theorem 4.7.]{Montacute2024} and show that they preserve reusability constraints on the universe of $\mathbb{PR}_{(k_1,k_2)}\mathcal{A}$.

  For the first direction, assume that $G(\mathcal{A})$ has a linear $(k_1,k_2)$-pebble linear component forest cover $(F, \overline{r}, p)$. 
  Let $F = P_1 \dotcup P_2 \dotcup \dotsc \dotcup P_n$ and for $i\in[n]$ let $P_i = v^i_1 \dotsc v^i_{m_i}$.
  We define the sequence $\overline{t}_i \coloneqq [(p(v^i_1), v_1),\dotsc,(p(v^i_{m_i}), v^i_{m_i})]$ and 
  $$\alpha_i \colon V(P_i) \to \mathbb{PR}_{(k_1,k_2)}\mathcal{A},\quad v^i_j \mapsto (\overline{t}_i, j).$$
  Since $V(\mathcal{A}) = \dotcup_{i=1}^n V(P_i)$ we can define $\alpha \colon V(\mathcal{A}) \to V(\mathbb{PR}_{(k_1,k_2)} \mathcal{A})$ 
  by setting $\alpha(v^i_j) \coloneqq \alpha_i(v^i_j)$ for $v^i_j \in V(P_i)$.
  Then $\alpha$ is well-defined because in $[(p(v^i_1), v_1),\dotsc,(p(v^i_{m_i}), v^i_{m_i})]$ no pebble $y_j \in [y_{k_2}]$ occurs more than once as pebble index in $P_i$
  by the definition of a linear $(k_1,k_2)$-pebble component forest cover. 
  The verification that $\alpha$ is a homomorphism and satisfies the counit-coalgebra and comultiplication-coalgebra axioms is verbatim to the corresponding statement in the proof of \cite[Theorem 4.7.]{Montacute2024}. 

  For the converse, let $\mathcal{A}$ together with $\alpha \colon \mathcal{A} \to \mathbb{PR}_{(k_1,k_2)}\mathcal{A}$ be a coalgebra. 
  For each $\overline{s}$ in the image of $\alpha$ we define the set 
  $P_{\overline{s}} \coloneqq \{a \in V(\mathcal{A}) : \exists i \in \mathbb{N}~\alpha(a) = (\overline{s}, i)\}$
  and an order on this set by setting 
  $a \leq_{\overline{s}} a'$ exactly if $\alpha(a)=(\overline{s},i)$, $\alpha(a')=(\overline{s},i'),$ and $ i \leq i'$.
  Furthermore, we define the pebbling function $p \colon V(\mathcal{A}) \to [x_{k_1}, y_{k_2}]$ by setting $p \coloneqq \pi_{\mathcal{A}} \circ \alpha$.
  By the counit-coalgebra axiom, we have $\varepsilon_{\mathcal{A}} \circ \alpha = \id_{\mathcal{A}}$, which is just that for every $a \in V(\mathcal{A})$
  there exists a unique $(\overline{s}, i) \in \mathbb{PR}_{(k_1,k_2)}$ such that $\alpha(a) = (\overline{s}, i)$ and $\overline{s}[i] = (z_i, a)$. 
  Also, by the comultiplication-coalgebra axiom we have that $\id_{\mathbb{PR}_{(k_1,k_2)}\mathcal{A}}^\ast \circ \alpha = (\alpha \circ \varepsilon_\mathcal{A})^\ast \circ \alpha$
  which just states that for every $(\overline{s}, i) \in \mathbb{PR}_{(k_1,k_2)}\mathcal{A}$ if $\overline{s}[i] = (z_i,a)$ we have that $\alpha(a) = (\overline{s}, i)$. 
  Together this yields that the elements in $P_{\overline{s}}$ are linearly ordered by $\leq_{\overline{s}}$ according to their position in $\overline{s}$
  and the sets $P_{\overline{s}}$ partition $V(\mathcal{A})$. 
  Now let $a \in P_{\overline{s}}$ such that $p(a) = y_j$ and let $\alpha(a) = (\overline{s}, i)$, then $y_j = z_i$.
  Now suppose that there exists $a' \in P_{\overline{s}}$ with $a' \neq a$ and $p(a') = y_j$. 
  Then we have $\alpha(a') = (\overline{s}, i')$ with $i' \neq i$ and thus $y_j = z_i \neq z_{i'} = y_j$, contradicting the definition of $\mathbb{PR}_{(k_1,k_2)}\mathcal{A}$. 
  It remains to prove that the linearly ordered sets $(P_{\overline{s}}, \leq_{\overline{s}})$ with $p$ induce a $(k_1+k_2)$-pebble linear forest cover of $G(\mathcal{A})$.
  This however is again shown in the proof of \cite[Theorem 4.7.]{Montacute2024} for the same construction.    

  \Cref{coalgebras:trees}: We review the constructions from \cite[Theorem 6]{Abramsky2017} and show that they preserve reusability constraints and depth. 

  Let $(F, \overline{r}, p)$ be a $(k_1,k_2)$-pebble forest cover of depth $q$ for $G(\mathcal{A})$. 
  Define $\alpha \colon \mathcal{A} \to \mathbb{P}^q_{(k_1,k_2)}$ by induction on $\preceq_F$: 
  For each $r_i \in \overline{r}$ we define $\alpha(r_i) \coloneqq [(p(r_i), r_i)]$ and for $a \in V(\mathcal{A}) \setminus \overline{r}$ let $a^-$ be the predecessor 
  of $a$ in the order $\preceq_T$. Then we set $\alpha(a) \coloneqq \alpha(a^-) [(p(a), a)]$.
  Then for each $a \in V(\mathcal{A})$ the length of the sequence $\alpha(a)$ is depth of $a$ in the forest $F$.
  Since also no pebble $y_j \in [y_{k_2}]$ is reused on a branch of $T$, we have $\alpha(a) \in \mathbb{P}^q_{(k_1,k_2)}\mathcal{A}$. 
  The verification of the coalgebra axioms is identical to the proof in \cite[Theorem 6]{Abramsky2017}. 

  Given a coalgebra $(\mathcal{A}, \alpha \colon \mathcal{A} \to \mathbb{P}^q_{(k_1,k_2)} \mathcal{A})$ over $\mathbb{P}^q_{(k_1,k_2)}$, we define a forest cover:
  As in the construction for the linear component forest cover, by the coalgebra axioms we have that 
  for each $a \in V(\mathcal{A})$ if $\alpha(a)=[(z_1,a_1),\dotsc,(z_n,a_n)]$, then $a=a_n$ and $\alpha(a_i) = [(z_1,a_1),\dotsc,(z_i,a_i)]$ for $i \in [n]$. 
  Also, there is a unique $\overline{s} \in \im(\alpha)$ such that the last entry of $\overline{s}$ is $(z, a)$.
  We set $a \preceq a'$ exactly if $\alpha(a) \sqsubseteq \alpha(a')$.
  The pebbling function is defined by setting $p(a)=z$ for $\overline{s}[(z,a)] \in \im(\alpha)$. 
  Then every branch of the resulting tree with tree order $\preceq$ hast length at most $q$, since this is the maximum length of a sequence $\alpha(a)$.
  Furthermore, each branch corresponds to prefixes of a sequence $\overline{s} \in \mathbb{P}^q_{(k_1,k_2)} \mathcal{A}$ so no pebble $y_j \in [y_{k_2}]$ occurs more than once. 
\end{proof}

Further employing the theory of game comonads, we give a categorical formulation of equivalence in the counting logics ${\land}\mathsf{C}^{(k_1,k_2)}_{\infty \omega}$ and $\mathsf{C}^{(k_1,k_2)}_{q}$
as isomorphism in the coKleisli category. Accordingly, similar theorems are also called \emph{isomorphism power theorems} in the literature.
As in \cite{Montacute2024}, we use the extended signature $\sigma^+ = \sigma \cup \{I\}$
and define the functor $J \colon \mathbf{Str}(\sigma) \to \mathbf{Str}(\sigma^+)$ which extends each $\sigma$-structure $\mathcal{A}$ by the identity relation $I^{J\mathcal{A}}= \{(a,a) : a \in V(\mathcal{A})\}$,
giving rise to a relative comonad on $J$. 

\begin{theorem} \label{thm:isomorphism}
    For all $\sigma$-structures $\mathcal{A}$ and $\mathcal{B}$ the following hold:
    \begin{enumerate}
        \item \label{isomorphism:paths} There exists a coKleisli isomorphism $\mathbb{PR}_{(k_1,k_2)} J \mathcal{A}\to J \mathcal{B}$ if and only if $\mathcal{A} \equiv_{{\land}\mathsf{C}^{(k_1,k_2)}_{\infty \omega}} \mathcal{B}$.
        \item \label{isomorphism:trees} There exists a coKleisli isomorphism $\mathbb{P}^q_{(k_1,k_2)} \mathcal{A}\to \mathcal{B}$ if and only if $\mathcal{A} \equiv_{\mathsf{C}^{(k_1,k_2)}_{q}} \mathcal{B}$.
    \end{enumerate}
\end{theorem}
\begin{proof}
  \Cref{isomorphism:paths}: We use \Cref{thm:logic:abp} and argue that exists a coKleisli isomorphism $\mathbb{PR}_{(k_1,k_2)} J \mathcal{A}\to J \mathcal{B}$
  if and only if Duplicator has a winning strategy for the game $\ABP^{(k_1,k_2)}(\mathcal{A}, \mathcal{B})$. 
  This follows from the the proof of \cite[Theorem 5.10.]{Montacute2024} by observing that there, Spoiler chooses the pebble sequence $\overline{z} = (z_1,\dotsc,z_n)$
  in the game exactly if the isomorphism $f \colon \mathbb{PR}_{(k_1,k_2)} J \mathcal{A}\to J \mathcal{B}$ maps sequences $\overline{s}$
  with pebble indexing $\overline{z}$ to such sequences again. This precisely shows that the reusability constraint in $\mathbb{PR}_{(k_1,k_2)}$ is preserved.  

  \Cref{isomorphism:trees}: Equivalence in the logic $\mathsf{C}^{(k_1,k_2)}_q$ is characterized by the \emph{bijective $(k_1,k_2)$-pebble game}, denoted as $\operatorname{BP}_{(k_1,k_2)}^q(\mathcal{A}, \mathcal{B})$,
  by \cite[Theorem 6]{Rassmann2025}. Thus, we again show that Duplicator has a winning strategy exactly if $\mathcal{A}$ and $\mathcal{B}$ are isomorphic in the coKleisli category and refer to \cite{Rassmann2025} for the definition of the game. 
  
  First, assume that there is an isomorphism $f \colon \mathbb{P}^q_{(k_1,k_2)} \mathcal{A} \to \mathcal{B}$ in the coKleisli category $\mathcal{K}(\mathbb{P}^q_{(k_1,k_2)})$.
  Then there is an inverse morphisms $g \colon \mathbb{P}^q_{(k_1,k_2)} \mathcal{B} \to \mathcal{A}$ such that $f \circ g^\ast = \id_{\mathbb{P}^q_{(k_1,k_2)} \mathcal{A}}$.
  Thus, for all $\overline{s} \in \mathbb{P}^q_{(k_1,k_2)} \mathcal{A}$ and $\overline{t} \in \mathbb{P}^q_{(k_1,k_2)} \mathcal{B}$ we have that 
  $\overline{t} = f^\ast(\overline{s})$ if and only if $\overline{s} = g^\ast(\overline{t})$. In particular, pebble sequences are preserved by $f^\ast$. 
  For each $\overline{s} \in \mathbb{P}^{q-1}_{(k_1,k_2)} \mathcal{A}$ and $z \in [x_{k_1}, y_{k_2}]$ we define the function $\psi_{\overline{s}, z} \colon V(\mathcal{A}) \to V(\mathcal{B})$
  by setting $\psi_{\overline{s},z}(a) = f(s[(z,a)])$.
  The winning strategy for Duplicator in $\operatorname{BP}_{(k_1,k_2)}^q(\mathcal{A}, \mathcal{B})$ is now to choose the bijection $\psi_{\overline{s},z}$ in each round 
  with position $\overline{s}$ and Spoiler picking up the pebble pair $z$. Since $f$ and $g$ are morphisms, every resulting position in the game will always induce a partial isomorphism.

  Conversely, when Duplicator has a winning strategy for $\operatorname{BP}_{(k_1,k_2)}^q(\mathcal{A}, \mathcal{B})$, we construct an isomorphism inductively as follows: 
  For every pair $\overline{s} \in \mathbb{P}^{q-1}_{(k_1,k_2)} \mathcal{A}, \overline{t} \in \mathbb{P}^{q-1}_{(k_1,k_2)} \mathcal{B}$ with the same pebble sequence 
  and $z \in [x_{k_1}, y_{k_2}]$ picked up by Spoiler, there is a bijection $\psi \colon V(\mathcal{A}) \to V(\mathcal{B})$ for Duplicator. 
  So setting $f^\ast(\overline{s}[(z,a)]) = \overline{t}[(z, \psi(a))]$ induces a bijective morphism $f^\ast \colon \mathbb{P}^q_{(k_1,k_2)} \mathcal{A} \to \mathbb{P}^q_{(k_1,k_2)}\mathcal{B}$
  and $\varepsilon_{\mathcal{B}} \circ f^\ast$ is the desired isomorphism. 
\end{proof}

One of the contributions of game comonads is to provide a unified language for various relations from finite model theory. 
Specifically, in the remainder of this section we show that morphisms in the coKleisli category characterize preservation and winning strategies for logics and games without counting.  
We first introduce the notion of reusability to the \emph{all-in-one pebble game} from \cite{Montacute2024} and the well-known \emph{existential $k$-pebble game} from \cite{Kolaitis1995}, allowing for a more fine grained analysis. 

\begin{definition}
    Let $\mathcal{A}, \mathcal{B}$ be $\sigma$-structures and $k_1+k_2 \in \mathbb{N}_+$. The \emph{all-in-one $(k_1,k_2)$-pebble game} $\AP^{(k_1,k_2)}(\mathcal{A}, \mathcal{B})$
    and \emph{$\exists\text{-}(k_1,k_2)$-pebble game} are defined as follows:
    Both games are played by the two players \emph{Spoiler} and \emph{Duplicator} on the structures $\mathcal{A}$ and $\mathcal{B}$
    and start from a (possibly empty) position $\overline{s}_0 \in ([x_{k_1}, y_{k_2}] \times V(\mathcal{A}))^m, \overline{d}_0 \in ([x_{k_1}, y_{k_2}] \times V(\mathcal{B}))^m$
    with the same pebble sequence in which each pebble pair occurs at most once. 

    In each round $n \in \mathbb{N}_+$ of the $\exists\text{-}(k_1,k_2)$-pebble game, the following steps are performed: 
    \begin{enumerate}
        \item Spoiler picks a pebble $z_n \in [x_{k_1}, y_{k_2}]$ such that $z_n$ is not yet placed or $z_n \in [x_{k_1}]$ and places it on an element $a_n \in V(\mathcal{A})$.
        \item Duplicator places the pebble $z_n$ on an element $b_n \in V(\mathcal{B})$.
    \end{enumerate}
    This induces sequences $\overline{s} = [(z_1,a_1),\dotsc,(z_n, a_n)]$ and $\overline{d} \coloneqq [(z_1,b_1),\dotsc,(z_n, b_n)]$ of placements after round $n$.

    In the single round of the game $\AP^{(k_1,k_2)}(\mathcal{A}, \mathcal{B})$, the following steps are performed:
    \begin{enumerate}
      \item Spoiler chooses a sequence $\overline{s} = [(z_1,a_1),\dotsc,(z_n, a_n)] \in ([x_{k_1}, y_{k_2}] \times V(\mathcal{A}))^n$ such that each $y_j \in [y_{k_2}]$ occurs at most once in $\overline{z} = (z_1,\dotsc,z_n)$
      and not in $\overline{s}_0$.
      \item Duplicator responds with a sequence $\overline{d} \coloneqq [(z_1,b_1),\dotsc,(z_n, b_n)] \in ([x_{k_1}, y_{k_2}] \times V(\mathcal{B}))^n$.
    \end{enumerate}
    The winning condition for both games is the following: 
    Duplicator wins the game if for all $i \in [m+n]$ the function $\eta_i$ defined by setting $\eta_i(\last_z(\overline{s}_0\overline{s}[1,i])) \coloneqq \last_z(\overline{d}_0\overline{d}[1,i])$
    for each $z \in [x_{k_1}, y_{k_2}]$ is a partial homomorphism between $\mathcal{A}$ and $\mathcal{B}$.
\end{definition}

The logic $\exists^+\mathsf{L}^{(k_1,k_2)}$ is defined as the fragment of existential positive first-order logic (i.e. no universal quantification and negation) over the variable set 
$[x_{k_1}, y_{k_2}]$ such that only variables from $[x_{k_1}]$ may be requantified. We define $\exists^+{\land}\mathsf{L}^{(k_1,k_2)}$ by additionally requiring that every conjunction is restricted 
and $\exists^+{\land}\mathsf{L}^{(k_1,k_2)}_q, \exists^+\mathsf{L}_q^{(k_1,k_2)}$ by bounding the quantifier-rank by~$q$. 

\begin{proposition} \label{prop:exist:games}
    For all $\sigma$-structures $\mathcal{A}$ and $\mathcal{B}$ the following hold:
    \begin{enumerate}
    \item \label{exist:games:ab} Duplicator wins $\AP^{(k_1,k_2)}(\mathcal{A}, \mathcal{B})$ if and only if $\mathcal{A} \Rrightarrow_{\exists^+{\land}\mathsf{L}^{(k_1,k_2)}} \mathcal{B}$.
    \item \label{exist:games:ex} Duplicator wins $q$ rounds of the $\exists\text{-}(k_1,k_2)$-pebble game if and only if $\mathcal{A} \Rrightarrow_{\exists^+\mathsf{L}_q^{(k_1,k_2)}} \mathcal{B}$.
    \end{enumerate}
\end{proposition} 
\begin{proof}
  \Cref{exist:games:ab}: First, we show by induction on $q \in \mathbb{N}$ that if there is a formula $\varphi \in \exists^+{\land}\mathsf{L}^{(k_1,k_2)}_q$ with
  $\mathcal{A}, \overline{s}_0 \models \varphi$ and $\mathcal{B}, \overline{d}_0 \centernot\models \varphi$, then 
  Spoiler has a winning strategy for the game $\AP^{(k_1,k_2)}(\mathcal{A}, \mathcal{B})$ from the starting position $\overline{s}_0, \overline{d}_0$ with a sequence of length $q$.
  Here we view $\overline{s}_0 = [(z_1,a_1), \dotsc, (z_m,a_m)], \overline{d}_0 = [(z_1,b_1), \dotsc, (z_m,b_m)]$ as variable assignments $z_i \mapsto a_i$. 
  For the base case $q=0$ the claim holds since mapping $a_i \mapsto b_i$ is not a partial homomorphism from $\mathcal{A}$ to $\mathcal{B}$
  exactly if Spoiler wins the $0$-move game. 
  For the inductive step, let $q>0$ and assume $\mathcal{A}, \overline{s}_0 \models \varphi$ and $\mathcal{B}, \overline{d}_0 \centernot\models \varphi$.
  By the same argument as for \Cref{prop:primitive} we may assume that $\varphi$ is a disjunction of primitive formulas. 
  If $\varphi$ is a restricted conjunction, there is at most one conjunct which is not a sentence and therefore contains all the free variables.
  Since $\mathcal{A}, \overline{s}_0$ and $\mathcal{B}, \overline{d}_0$ must disagree on this conjunct, 
  it suffices to consider $\varphi = \exists z_i \psi$ with $\psi \in \exists^+{\land}\mathsf{L}^{(k_1,k_2)}_{q-1}$. 
  Then there exists $a \in V(\mathcal{A})$ with $\mathcal{A}, \overline{s}_0[(z_i,a)] \models \psi$ but for all $b \in V(\mathcal{B})$ we have 
  $\mathcal{B}, \overline{d}_0[(z_i, a)] \centernot\models \psi$. 
  Thus, Spoiler has a winning strategy by placing $z_i$ on $a$ on the first move and proceeding by induction for the remaining game. 
  If $z_i \in [y_{k_2}]$ occurs in the pebble sequence of $\overline{s}_0$ we may remove the placement from $\overline{s}_0$ to avoid reusing the pebble pair in the strategy.

  For the converse direction, let $\overline{z} \in [x_{k_1}, y_{k_2}]^n$ and $\overline{a} \in \mathcal{A}^n$ and define for $\overline{za} = [(z_1,a_1), \dotsc (z_n,a_n)]$ and $i \in [n]$ the formulas
  \begin{align*}
    \mathsf{diag}_{\overline{zv}[1,i]} & \coloneqq \{R(z_{i_1},\dotsc, z_{i_r}) : R \in \sigma, r\in [n], \forall i \in [r]~\last_{z_{i_\ell}}(\overline{pv}[1,i]) = v_{i_\ell},  (v_{i_1},\dotsc,v_{i_r})\in R^{\mathcal{A}}\}. \\
    \varphi_0 & \coloneqq \top \\
    \varphi_i & \coloneqq \exists z_i \varphi_{i-1} \land \bigwedge_{\psi \in \mathsf{diag}_{\overline{zv}[1,i]}} \psi.
  \end{align*}
  If picking the sequence $\overline{za}$ is a winning move for Spoiler, then the sentence $\varphi_n$ holds on $\mathcal{A}$, but not on $\mathcal{B}$. 

  \Cref{exist:games:ex}: The proof of \cite[Theorem 4.8]{Kolaitis1995} already implicitly preserves quantifier-rank and reusability constraints. 
\end{proof}

Finally, we also obtain what is called a \emph{morphism power theorem}, linking morphisms of game comonads with restricted reusability to preservation in counting-free logics with restricted requantification. 
\begin{theorem} \label{thm:morphism}
    For all $\sigma$-structures $\mathcal{A}$ and $\mathcal{B}$ the following hold:
    \begin{enumerate}
        \item \label{morphism:path} There exists a coKleisli morphism $f \colon \mathbb{PR}_{(k_1,k_2)} \mathcal{A} \to \mathcal{B}$ if and only if $\mathcal{A} \Rrightarrow_{\exists^+{\land}\mathsf{L}^{(k_1,k_2)}} \mathcal{B}$ 
        \item \label{morphism:tree} There exists a coKleisli morphism $f \colon \mathbb{P}^q_{(k_1,k_2)} \mathcal{A} \to \mathcal{B}$ if and only if $\mathcal{A} \Rrightarrow_{\exists^+\mathsf{L}_q^{(k_1,k_2)}} \mathcal{B}$
    \end{enumerate}
\end{theorem}
\begin{proof}
  \Cref{morphism:path}: Assume Duplicator has a winning strategy for $\AP^{(k_1,k_2)}(\mathcal{A}, \mathcal{B})$.
  Then for every choice $\overline{s} = [(z_1,a_1),\dotsc,(z_n,a_n)] \in ([x_{k_1}, y_{k_2}] \times V(\mathcal{A}))^n$
  there exists a winning choice $\overline{d} = [(z_1,a_1),\dotsc,(z_n,a_n)] \in ([x_{k_1}, y_{k_2}] \times V(\mathcal{B}))^n$ for Duplicator. 
  We define $f^\ast \colon \mathbb{PR}_{(k_1,k_2)}\mathcal{A} \to \mathbb{PR}_{(k_1,k_2)}\mathcal{B}$ by setting 
  $f^\ast(\overline{s}, i) = (\overline{d}, i)$ for each $i \in [n]$. This coextension preserves pebble sequences and hence reusability constraints. 
  The function $\varepsilon_{\mathcal{B}} \circ f^\ast$ is the desired coKleisli morphism, see the proof of \cite[Theorem 5.5]{Montacute2024}. 

  Assume that there exists a coKleisli morphism $f \colon \mathbb{PR}_{(k_1,k_2)} \mathcal{A} \to \mathcal{B}$.
  Let Spoiler choose the sequence $\overline{s} = [(z_1,a_1),\dotsc,(z_n,a_n)] \in ([x_{k_1}, y_{k_2}] \times V(\mathcal{A}))^n$ in $\AP^{(k_1,k_2)}(\mathcal{A}, \mathcal{B})$,
  then for $f^\ast(\overline{s}, n) = (\overline{d}, n)$ with  $\overline{d} = [(z_1, f(\overline{s},1)),\dotsc, (z_n, f(\overline{s}, n))]$
  let Duplicator choose $\overline{d}$ as response. This in particular respects the pebble sequence chosen by Spoiler. 
  It follows that setting $\eta_i(\last_z(\overline{s}[1,i])) \coloneqq \last_z(\overline{d}[1,i])$ for each $z \in [x_{k_1}, y_{k_2}]$ is a partial homomorphism between $\mathcal{A}$ and $\mathcal{B}$
  from the active pebble and compatibility conditions. 

  \Cref{morphism:tree}: All constructions in the proof of \cite[Theorem 13]{Abramsky2017} preserve pebble sequences and thus also the number of rounds and reusability. 
\end{proof}
  \section{Conclusion}
In this work, we extended the analysis of homomorphism indistinguishability to graph classes characterized by graph decompositions with restricted reusability. 
We demonstrate how decomposition-based approaches offer robust and adaptable techniques for characterizing indistinguishability relations as well as for establishing homomorphism distinguishing closedness.
Moreover, by integrating these results within the broader framework of game comonads, we present a unified categorical perspective on the role of requantification in finite variable counting logics.
We list some open questions for future work: 
\begin{itemize}
    \item The homomorphism indistinguishability relation $\equiv_{\mathcal{T}^{(k_1,k_2)}_q}$ can be decided by a more space-efficient variant of the $(k_1+k_2)$-dimensional Weisfeiler-Leman algorithm \cite{Rassmann2025}. 
    It would be interesting to develop a \emph{linearized variant} to efficiently decide the relation $\equiv_{\mathcal{P}^{(k_1,k_2)}}$.  
    The logic $\exists^+ \mathsf{L}^k$ is closely related to the \emph{$k$-consistency algorithm} for solving \emph{constraint satisfaction problems},  
    so it might be fruitful to explore whether restricting reusability yields improved algorithmic techniques. 
    \item Motivated by the constructive nature of our results, it is natural to ask for more connections between model-comparison and pursuit-evasion games. 
    Specifically, identifying broader classes of games that align with logical equivalences could yield more results on characterizations and h.d. closedness.
    \item While we have established h.d. closedness for specific graph classes,  
    a comonadic treatment of this property could provide a deeper understanding of its categorical structure.  
    In particular, investigating whether h.d. closedness can be characterized via coalgebraic properties of game comonads 
    might reveal fundamental principles governing homomorphism indistinguishability of relational structures as in \cite{Dawar2021}. 
\end{itemize} 
  \bibliography{bibliography}
\end{document}